\documentclass[preprint,12pt]{elsarticle}

\usepackage{mathtools}
\usepackage[ruled,vlined]{algorithm2e}
\usepackage{subcaption}
\usepackage{algorithmic}
\usepackage{footmisc}
\def\changesHilighted{true}
\usepackage{amsmath}
\usepackage{graphicx,epstopdf}
\usepackage{natbib}
\usepackage{color}
\usepackage{amsfonts}
\usepackage{stfloats}
\usepackage{bm}
\usepackage[utf8]{inputenc}
\usepackage[]{footmisc}
\usepackage{amsthm}
\usepackage{amsmath,amsthm,amsfonts,amssymb,amscd}
\usepackage{lastpage}
\usepackage{enumerate}
\usepackage{fancyhdr}
\usepackage{mathrsfs}
\usepackage{xcolor}
\usepackage{graphicx}
\usepackage[english]{babel}
\usepackage[pdfencoding=auto]{hyperref}
\usepackage{bookmark}
\usepackage{graphicx}  
\usepackage{amssymb}
\usepackage{amsmath}
\usepackage{amsthm}

\newtheorem{problem}{Problem}
\newtheorem{theorem}{Theorem}

\newcommand{\norm}[1]{\left\lVert #1 \right\rVert}

\definecolor{dogwoodrose}{rgb}{0.84, 0.09, 0.41}

\definecolor{slateblue}{rgb}{0.42, 0.35, 0.8}
\ifnum\pdfstrcmp{\changesHilighted}{false}=0
    
\fi
\newcommand{\R}{\mathbb{R}}

\newcommand{\C}{\mathcal{C}}

\newtheorem{remark}{Remark}
\newtheorem{assumption}{Assumption}
\newtheorem{definition}{Definition}
\newtheorem{lemma}{Lemma}

\usepackage{amssymb}
\usepackage{amsthm}
\usepackage{mathtools}
\usepackage{url}
\usepackage{xurl} 
\usepackage{graphicx}
\usepackage[export]{adjustbox} 
\usepackage{float}
\usepackage{multirow}
\usepackage{makecell}
\usepackage{color}
\usepackage{soul}
\usepackage{amsmath}
\journal{Advanced Engineering Informatics}

\begin{document}

\begin{frontmatter}

\title{TRUST-UP: Trustworthy Reinforcement learning Using Safe Techniques for UAV Pursuit}

\author[1,2]{Yaosheng Deng\corref{equal}}
\ead{yaosheng001@e.ntu.edu.sg}
\author[2,3]{Mengtao Lyu\corref{equal}}
\ead{mengtao.lyu@gatech.edu}

\author[1]{Junjie Gao}
\ead{junjie008@e.ntu.edu.sg;}

\author[1]{Jiaping Xiao}
\ead{jiaping001@e.ntu.edu.sg}

\author[1,2]{Mir Feroskhan\corref{cor1}}
\ead{mir.feroskhan@ntu.edu.sg}

\cortext[cor1]{Corresponding author}
\cortext[equal]{These authors contributed equally to this work.}

\affiliation[1]{organization={School of Mechanical and Aerospace Engineering, Nanyang Technological University},
            % addressline={Address One}, 
            city={Singapore},
            postcode={639798}, 
            % state={State One},
            country={Singapore}}
\affiliation[2]{organization={Air Traffic Management Research Institute, Nanyang Technological University},
            % addressline={Address Two}, 
            city={Singapore},
            postcode={637460}, 
            % state={State Two},
            country={Singapore}}
\affiliation[3]{organization={Georgia Institute of Technology},
            % addressline={Address Two}, 
            city={Atlanta},
            postcode={30332}, 
            state={GA},
            country={USA}}

%% Abstract
\begin{abstract}
Reinforcement Learning (RL) enables autonomous aerial vehicles to adapt quickly and make efficient decisions, making it well-suited for dynamic urban air mobility operations. However, the lack of safety guarantees and transparency hinders the airworthiness certification of RL-based flight control systems, particularly in low-altitude urban environments with human presence. This paper proposes a trustworthy reinforcement learning algorithm that utilizes safe techniques to address the AI trustworthiness requirements for aviation safety, ensuring the transparent and certifiable deployment of RL in safety-critical aerial operations. Specifically, we proposed a Trustworthy Reinforcement learning Using
Safe Techniques for UAV Pursuit (TRUST-UP), which consists of two key components: a safety filter constructed from Control Barrier Functions (CBFs) that transforms unsafe RL actions into provably safe flight commands, and a switching strategy that enhances feasibility while maintaining operational transparency. These components enable trustworthy AI deployment in urban airspace, satisfying technical robustness and transparency requirements for aviation certification.
Simulation results demonstrate that TRUST-UP enables autonomous UAVs to safely navigate congested urban environments while maintaining human-interpretable decision logic. This work contributes toward certifiable and explainable AI frameworks for low-altitude aviation, addressing the critical need for trustworthy autonomous flight systems in future urban air mobility.
\end{abstract}

% \begin{highlights}
% \end{highlights}

%     \includegraphics[width = 1.05\textwidth]{Figures/Fig 0.png}

\end{frontmatter}


%% 1. Introduction
\section{Introduction}
\label{sec: 1-Introduction}
    %% Sec. 1 Introduction
% Motivation
Autonomous pursuit and navigation are fundamental capabilities for next-generation aerial vehicles, with critical applications in urban air mobility including eVTOL emergency response operations~\cite{shafiqurrahman2023electric,qiu2025electric}, UAV-based traffic monitoring in urban corridors~\cite{de2019multi,shao2021real}, and coordinated air-ground surveillance systems~\cite{zhang2022game,peng2025air}. Reinforcement Learning (RL) has demonstrated superior performance in enabling agile flight control and adaptive navigation for these aerial missions~\cite{de2021decentralized,wang2023high}. However, the lack of safety guarantees and transparency hinders the airworthiness certification of RL-based flight control systems, particularly in low-altitude urban environments with human presence, violating the stringent requirements for trustworthy autonomous systems. This fundamental deficiency stems not just from technical limitations, but from a failure to address the core issue of human perception.

Crucially, the challenge of deploying autonomous systems in human environments overlooks the vital concept of perceived safety, which is a critical factor for trustworthiness and public acceptance. Current RL policies primarily address technical safety, focusing on metrics like collision avoidance with inanimate objects. This narrow focus fails to account for how safety is perceived by people. Unlike inanimate objects, humans require expanded "trust radii" that reflect psychological comfort, creating a fundamental conflict between machine optimized behavior and trust from a human perspective~\cite{zhou2025knowledge}. This conflict manifests in multiple ways. In a broad sense, a drone's mere presence over a crowded plaza can induce widespread anxiety, as it is often perceived as an unpredictable risk regardless of its adherence to a legally safe altitude~\cite{Torija2021, Wing_complaints_reference_2022, liu2021deep, gu2024review}. The conflict becomes more acute and personal in close proximity encounters, as illustrated in Figure~\ref{fig:diyizhang}. An RL UAV may execute a high speed, last second maneuver to narrowly avoid a pedestrian. From a purely technical standpoint, this action is a success because a collision is averted. From the perspective of human perception, however, the maneuver is a failure. The aggressive trajectory invades the human's peripersonal space, which is a cognitive safety zone, and is perceived as a dangerous loss of control which shatters trust~\cite{Ferroni2022, Chidambaram2022}. Therefore, a system that is technically safe can simultaneously be perceived as dangerously untrustworthy. These real-world psychological responses expose a critical gap between technical safety and perceived trustworthiness. The concept of a ``trust radius'' is grounded in the established principle of proxemics, which defines the personal ``safety space'' that individuals unconsciously maintain around themselves~\cite{Hall1966}. This finding is supported by numerous recent UAM acceptance studies~\cite{FAAUAMConOps2020}. Recognizing this, major aviation authorities like the FAA and EASA have formally identified ``societal acceptance'' as a key certification barrier for autonomous systems~\cite{Eker2020}. The challenge of respecting these psychological boundaries is where existing safe RL methods reveal their fundamental inadequacy.

While the field of safe RL has emerged to provide safety assurances, existing methods focus on general safety barriers for autonomous flight systems and are ill-equipped to handle the dynamic ``trust radii'' humans require. For example, reward shaping techniques were proposed to penalize unsafe flight behaviors~\cite{yang2023model,marvi2021safe}. {However, these constrained RL techniques model safety as soft penalties, which cannot provide the deterministic guarantees required for human psychological trust.} The traditional shield-based methods were introduced to filter RL actions through safety layers~\cite{carr2023safe}, but they lack the specificity to guarantee safety around pedestrians. Even more advanced Control Barrier Function (CBF) integration with RL, proposed for flight safety certification~\cite{cohen2023safe,wang2022ensuring}, fails to guarantee feasibility when enforcing human-aware constraints. Specifically, predictive-based approaches assume uniform safety distances~\cite{hu2023safe}, relaxation-based methods contradict the strict safety margins mandated for human protection~\cite{chen2023safety}, and learning-based synthesis adapts constraints using static obstacle data~\cite{so2024train, chen2024learning}, none of which can capture the psychological factors of human comfort zones. This inadequacy is compounded by the black-box nature of RL, which violates transparency requirements by making it impossible to verify how an agent will navigate these boundaries~\cite{soudain2025easa}.
% {Furthermore, existing CBF-RL frameworks typically impose rigid, constant input constraints, which induce infeasibility when navigating complex environments~\cite{deng2025safety,he2025state,zhang2025control}. 
% Constrained by these static limits, the UAV cannot satisfy the ``trust radii'' that humans require.}
These challenges become acute in aerial pursuit missions, which demand simultaneous satisfaction of conflicting constraints of sensing range, collision avoidance, and thrust limitations, making it clear that a new framework is needed to bridge this gap.

% These real-world psychological responses expose a critical gap: existing frameworks lack a mechanism to reconcile technical safety with perceived safety. Our work aims to bridge this gap by creating a system that co-optimizes for both provable collision avoidance (technical safety) and adherence to human comfort zones (perceived safety).
% The black-box nature of RL compounds these issues, preventing air traffic controllers from anticipating or explaining UAV behaviors to concerned stakeholders, directly violating aviation transparency requirements~\cite{soudain2025easa}. Without formal guarantees for human-aware safety margins and verifiable control logic, RL-based systems cannot achieve the operational certification and public acceptance essential for urban air mobility~\cite{liu2021deep,gu2024review}.

To address this critical gap between ensuring human-aware safety and maintaining mathematical feasibility, this paper proposed a Trustworthy Reinforcement learning Using Safe Techniques for UAV Pursuit (TRUST-UP) algorithm that provides formal safety guarantees with transparent decision logic for autonomous aerial vehicles. The TRUST-UP includes a trained RL model for aerial pursuit, a safety filter constructed from three adaptive CBF constraints with provable feasibility, and a switch strategy that provides interpretable safety decisions between RL outputs and filtered flight commands. To design our algorithm, we first augment the UAV dynamics by introducing virtual control inputs, which transform the thrust-constrained problem into an output-constrained problem, enhancing the feasibility of the associated QP. {This transformation resolves the potential infeasibility induced by traditional static input constraints~\cite{deng2025safety,he2025state,zhang2025control}, ensuring the satisfaction of human-aware safety constraints during highly dynamic interactions.} Based on the UAV's position and velocity, we then construct the safety filter using three adaptive CBFs, ensuring that unsafe RL outputs are converted into safe flight commands even under wind disturbances. Finally, we incorporate a switch strategy to alternate between RL actions and the safety filter, and we prove that the proposed switch strategy ensures the safety filter satisfies the Karush-Kuhn-Tucker conditions for each safety constraint, guaranteeing the airworthiness compliance of TRUST-UP. The main contributions of this paper are summarized as follows:
\begin{itemize}
\item We address thrust constraints by augmenting the UAV system with virtual control inputs and designing a thrust-constrained CBF based on the augmented framework. This approach, unlike conventional constant constraints~\cite{deng2025safety,sun2024moving}, enforces operational flight envelope limits while adaptively relaxing upper bounds during emergency maneuvers, maintaining mission performance while ensuring verifiable flight safety.
\item We design two position-based CBFs to ensure separation standards and sensor coverage in urban airspace. Together with the thrust-constrained CBF, these form a safety filter solved as a QP with guaranteed feasibility. The adaptive technique enables UAVs to maintain safe operations even under atmospheric disturbances.
\item We proposed the TRUST-UP algorithm with a transparent switch strategy that determines RL output safety and switches accordingly. Unsafe RL outputs are converted into safe flight commands by solving a QP. We formally prove that TRUST-UP satisfies the KKT conditions for all safety constraints, ensuring compliance with aviation trustworthiness requirements for certified deployment.
\end{itemize}
\begin{figure}[!ht]
    \centering
    \includegraphics[width=1\linewidth]{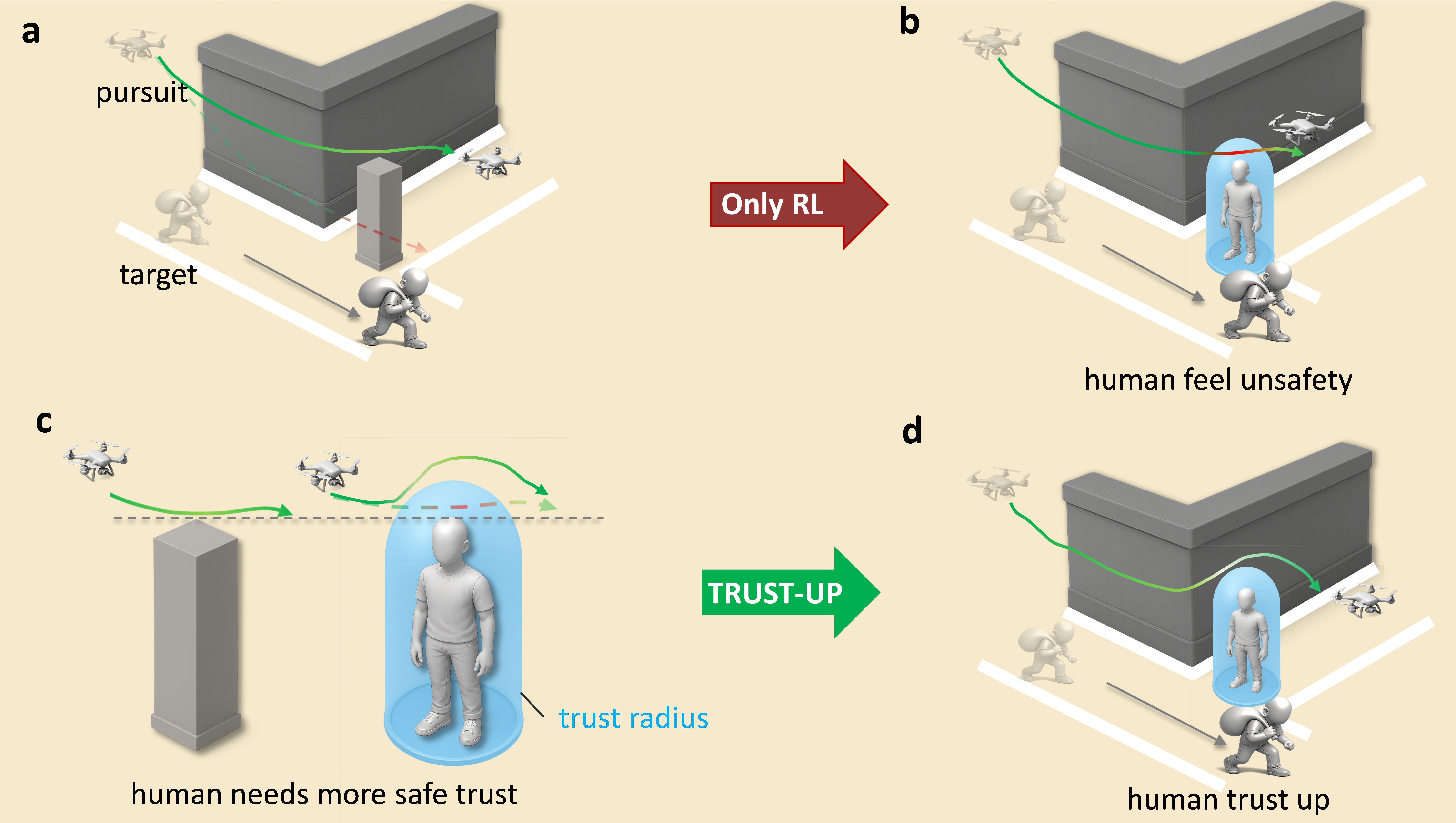}
    \caption{The trust radius challenge in deploying RL-based UAVs in human-populated environments and the TRUST-UP solution. {a}, Conventional RL training scenario with static obstacles, where UAVs learn pursuit behaviors using uniform safety distances suitable for inanimate objects. {b}, Deployment failure when the same RL policy encounters humans, where inadequate safety margins lead to {perceived lack of safety} and may violate regulatory requirements. {c}, Fundamental safety requirement disparity: humans require an expanded {required trust radius} (blue region) compared with inanimate obstacles because of psychological, regulatory, and ethical considerations that standard RL training cannot capture. {d}, TRUST-UP framework enforcement of human-aware safety boundaries for {enhanced human trust}.
}
    \label{fig:diyizhang}
\end{figure}
The rest of the paper is organized as follows. Section~\ref{sec-preliminary} presents some preliminaries. Section~\ref{sec-ps} formulates the aerial pursuit problem, including UAV dynamics, airspace safety constraints, and mission objectives. Section~\ref{sec-method} presents the proposed TRUST-UP algorithm, detailing the thrust-constrained CBF, the safety filter, and the switch strategy. The theoretical guarantees for aviation certification are also provided in this section. Section~\ref{sec-experiment} validates the proposed method through simulations, demonstrating its effectiveness in ensuring flight safety and performance in various urban air mobility scenarios. Finally, Section~\ref{sec-conclusion} concludes the paper and discusses potential directions for future work.

%% 2. Related work
\section{Preliminary}
\label{sec-preliminary}
    Let $\mathbb{R}$, $\mathbb{R}^+$ denote the field of real numbers and the set of non-negative reals. For a vector $x \in \mathbb{R}^n$, 
$x_i$ denotes the $i$-th element of $x$, $\|x\| = \sqrt{x_1^2 + \dots + x_n^2}$ denotes the two-norm of $x$.
Let $M$ be a matrix, define $\lambda_{\min}M$ as the minimum eigenvalue of a matrix $M$. A continuous function $\alpha:[0,a)\to[0,\infty)$ is class- $\mathcal{K}$ for some $a>0$ if it is strictly increasing on the domain, and $\alpha(0)=0$. It is class-$\mathcal{K}_\infty$ if $\lim_{r\to\infty}\alpha(r)\to\infty$. 
\subsection{CBF}\label{sec-preliminary-a}
Consider the following control affine system
\begin{equation}\label{preliminary SDE}
\begin{aligned}\dot{x}(t)=f(x) + g(x)u(t),  \end{aligned}
\end{equation}
where, $x\in\R^n$ is the state variable, $u\in\R^m$ denotes the control input of~\eqref{preliminary SDE}. $f$ and $g$ are locally Lipschitz continuous. 
\begin{definition}\label{def-Hs}~\cite{ames2019control}
Let 
\begin{equation}\label{prelimi-set H}
\C=\{x\in\R^n:h(x)\geq0\},
\end{equation}
where $h\colon\R^n\to\R$ is a continuously differentiable function.
Set $\C\subset \R^n$ is forward invariant for system~\eqref{preliminary SDE} if its solutions starting at all $x(t_0)\in\C$ satisfy $x(t)\in\C$ for all $t\geq t_0$.
\end{definition}
\begin{definition}\label{def-cbf}~\cite{ames2020integral}
Given a set $\C$ as in~\eqref{prelimi-set H}, $h(x)$ is a CBF for system~\eqref{preliminary SDE} if there exists a class $\mathcal{K}$ function $\gamma$ such that 
\begin{equation}
L_fh(x)+L_gh(x)u + \gamma(h(x))\geq 0
\end{equation}
for all $x\in\C$.
\end{definition}
\subsection{High-order CBF (HOCBF)}
In the context of high-order control barrier functions, first, we introduce relative degree:
\begin{definition}
A continuously differentiable function $h$ is said to have relative degree $r$ on a given domain of $x$ with respect to system~\eqref{preliminary SDE} if $L_gL_f^kh(x)=0$ for all $k<r-1$ and $L_gL_f^kh(x)\ne 0$ hold for all $x\in\R^n$.
\end{definition}
While Definition 2 is only applicable to CBFs with a relative degree of one, many applications often involve CBFs with higher relative degrees. To accommodate such scenarios, an extended definition known as HOCBFs has been developed in~\cite{intro-multicbf-3}. In this approach, we define a series of continuously differentiable functions $\hbar\colon \R^n\to \R$ as
\begin{equation}\label{zeta}
\begin{aligned}&\hbar_1({x})=h({x}),\\&\hbar_i({x})=\dot{\hbar}_{i-1}({x})+\alpha(\zeta_{i-1}({x})), i\in\{1,\ldots,r\},\end{aligned}
\end{equation}
We also define their zero-superlevel sets $\C_i$ and their interior sets for $i\in\{1,...,r\}$ as
\begin{equation}\label{zeta set}
    \begin{aligned}\mathcal{C}_i=\{{x}\in\mathbb{R}^n\mid\hbar_{i-1}({x})\geq0\}.\end{aligned}
\end{equation}
\begin{definition}
Let the functions $\zeta_i(x)$ and sets $\mathcal{C}_i$ be defined by~\eqref{zeta} and~\eqref{zeta set}, respectively. The $r$-th order continuously differentiable function $h(x)$ with relative degree $r > 1$ is called a HOCBF if $h$ and its derivatives up to order $r$, are locally Lipschitz continuous such that 
\begin{equation}
\begin{aligned}\sup_{u\in\mathbb{R}^m}\Big[L_f^rh+L_gL_f^{r-1}h {u}+\sum_{i=1}^{r-1}L_f^i(\alpha\circ\hbar_{r-i-1}).\\+\alpha(\hbar_{r-1})\Big]\geq0,\end{aligned}
\end{equation}
for all ${x}\in\bigcap_{i=1}^r (\mathcal{C}_i)$.
\end{definition}

%% 3. VALIO
\section{Problem Statement}
\label{sec-ps}
    This paper aims to design a safe algorithm for the target pursuing problem with collision avoidance, sensing limitation, and input constraint. The algorithm design is based on the following safe strategy. 
\subsection{Safe Strategy for Pursuers}\label{subsec-ps-1}
Consider a group of $n$ pursuit UAVs and $n$ target UAVs. Let $x_i\in\R^n$ where $i\in\mathcal{I}_x = \{1, \ldots, N\}$ be the position of the $i$-th pursuit UAV and $\mathcal{X}=\{x_i\}$ denotes a
collection of functions of $x_i$. Let $q_i\in\R^n$, $i\in\mathcal{I}_q=\{1, \ldots, N\}$, be the position of the $i$-th target UAV and $\mathcal{Q}=\{q_i\}$ denotes a collection of functions of $q_i$. Note that the $i$-th pursuer's primary task is tracking the corresponding target position $q_i$. Suppose the sensing range of each pursuit UAV is $R_i\in\R^+$, then the sensing safe for $i$-th pursuer should satisfy
\begin{equation}\label{xi-qi<Ri}
\norm{p_i(t)-q_i(t)}\leq R_i,
\end{equation}
for all $i\in\mathcal{I}_x$ and $t\geq 0$.

Similar to the UAVs, the environment obstacles $j\in\mathcal{I}_o=\{1, \ldots, M\}$ are also modeled as a rigid sphere at $o_j\in\mathcal{O}\subset \R^{n,M}$. Let $\mathcal{P}=\mathcal{X}\cup\mathcal{Q}\cup\mathcal{O}$ denote the set of all UAVs and obstacles. To describe the $i$-th pursuit UAV's collision avoidance requirements, define $\mathcal{P}_{-i}=\mathcal{P} \backslash\{x_i\}$ as the set of all UAVs excluding $x_i$. The safe distance between the position of $i$-th pursuer $x_i$ and the position $p_k\in \mathcal{P}_{-i}$ should satisfy
\begin{equation}\label{xi-pk>ri}
\norm{x_i(t) - p_k(t)} \geq  r_i,
\end{equation}
for all $k\in \mathcal{I}_{k}=\{1,\ldots,2N+M-1\}$ and all $t\geq 0$. 

{The proposed trust radii in~\eqref{xi-qi<Ri} and~\eqref{xi-pk>ri} represent the maximum acceptable bound of human proxemics, thereby establishing a deterministic and interpretable safety baseline for human-populated environments.} {The following problem formulation is established under a synchronized shared-state information setting, where each pursuer has access to synchronized fused state estimates of the relevant agents and obstacles needed to construct the set $P^{-i}$.
}
Next, we made the following assumption for the observation of $p_k$.
\begin{assumption}\label{asm-p dot p}
The velocity and acceleration of the $k$-th UAV, $\dot p_k$ and $\ddot p_k$, are bounded by two positive constraint $\rho_v\in\R^+$ and $\rho_a\in\R^+$, such that $\norm{\dot p_k(t)}\leq \rho_v$, and $\norm{\ddot p_k(t)}\leq \rho_a$ for all $k\in\mathcal{I}_{k}$ and $t\geq 0$.
\end{assumption}
% Note that the $i$-th pursuit UAV's primary task is tracking the corresponding target $q_i$ while avoiding collisions with all obstacles $\mathcal{O}$ and other pursuit and target UAVs.
% Let $\mathcal{I} = \{0, 1, \ldots, N\}$ represent the set of all UAVs in the environment, excluding the pursuit UAV. Collectively, the positions of all these UAVs are denoted as $p_i\in\R^n$, $i \in \mathcal{I}$. The UAV with index $0$ is designated as the target for pursuing, with its position denoted by $p_0 \in \R^n$. The positions of the remaining UAVs, considered as static or dynamic obstacles, are represented as $p_j \in \R^n$, $j \in \mathcal{J} = \{1, \ldots, N\}$.  
% The following assumption is made regarding $p_i$.
\subsection{Statement of Problem}
The following control affine system governs the position vector  $x_i$ of the $i$-th pursuit UAV:
\begin{equation}\label{sys1}
\begin{aligned}
\dot{x}_i(t)=f(x_i) + g(x_i)u_i(t) + Y(x_i)\theta,
\end{aligned}
\end{equation}
where $x_i\in\R^n$ is the state of $i$-th pursuer, $u_i\in\R^m$ is its control input. $f(x_i)$ is the drift term, and $g(x_i)$ is the control matrix, both of which are locally Lipschitz continuous functions of $x_i$. $Y:\mathbb{R}^n\to\mathbb{R}^{n\times p}$ represents known locally Lipschitz lumped uncertainties, and $\theta \in \R^p$ is a constant vector of uncertain parameters. 

Let $u_{\text{nom},i}$ denote the nominal control input of system~\eqref{sys1}, which is designed to achieve the primary objective of system~\eqref{sys1}, such as target tracking or stabilization. It represents the control input in the absence of safety constraints and governs the behavior of the $i$-th pursuer as
\begin{equation}\label{syscloseloop}
\begin{aligned}
\dot{x}_i(t)=f(x_i) + g(x_i)u_{\text{nom},i}(t) + Y(x_i)\theta,
\end{aligned}
\end{equation}
One specific form of $u_{\text{nom},i}$ can be:
\begin{equation}
u_{\text{nom},i}(t) = \int_0^t \pi_i(x_i(\tau))d\tau,
\end{equation}
where $\pi_i(x_i)$ is a continuous control law derived, for example, from model-free RL~\cite{zhang2023adaptive} or optimization-based methods~\cite{jiang2023incorporating}. However, directly applying $u_{\text{nom},i}$ in system~\eqref{syscloseloop} can lead to unsafe behaviors, such as violating input limits, collision constraints, or sensing requirements. 
To mitigate these risks, we first introduce a time-varying input constraint $\kappa=\kappa(x_i(t))$ for the $i$-th pursuer to regulate its control input $u_i$
\begin{equation}\label{adaptive bound}
\norm{u_i(t)}\leq \kappa(x_i(t)),
\end{equation}
for all $t\geq 0$. {This adaptive bound $\kappa(x_i(t))$ in~\eqref{adaptive bound} correlates the maximum thrust allowance with the system state $x_i(t)$. Through~\eqref{adaptive bound}, the UAV can adaptively relax its maximum control authority to guarantee human trust radii during sudden evasive actions, while strictly regulating erratic RL commands during nominal flight.} To enforce this input constraint, we define an input constraint safe set for system~\eqref{sys1} using the CBF technique. Specifically, a Lipschitz continuous function $h_{u,i}$ is defined as an input constraint CBF.
\begin{equation}\label{hu}
h_{u,i}(x_i,u_i,t) = \kappa(x_i(t))^2 - \norm{u_i(t)}^2,
\end{equation}
and the input constraint safe set for the $i$-th pursuer is defined as: \begin{equation}\label{C3} \C_{u,i} = \{u_i \in \mathbb{R}^m \mid h_{u,i}(x_i,u_i,t) \geq 0\}. \end{equation}

Unlike conventional single-target pursuit problems, this paper considers $N$ target-pursuer pairs in an environment with $M$ obstacles. In such a complex scenario, each pursuer must ensure sensing range safety as in~\eqref{xi-qi<Ri} while simultaneously maintaining collision avoidance by satisfying the safe radius condition in~\eqref{xi-pk>ri}. These dual safety requirements necessitate the formulation of appropriate safety constraints. Therefore, we introduce the collision avoidance safety set $\C_{c,i}$ and sensing range safety set $\C_{s,i}$ for each pursuer $i$ as follows:
\begin{equation}\label{C1}
C_{c,i} = \bigcap_{k\in\mathcal{I}_k}\{x_i\in\R^n\mid h_{c,i,k}(x_i(t))\geq 0\},
\end{equation}
\begin{equation}\label{C2} 
\C_{s,i} = \{x_i\in\R^m\mid h_{s,i}(x_i(t))\geq 0\},\end{equation}
where  $h_{c,i,k}$ and $h_{s,i}$ are two Lipschitz continuous function with a relative degree $2$ defined as
\begin{equation}\label{h1}
h_{c,i,k}(t) = \norm{x_i(t) - p_k(t)}^2 - r_i^2,
\end{equation}
\begin{equation}\label{h2} 
h_{s,i}(t) = R_i^2 - \norm{x_i(t) - q_i(t)}^2.
\end{equation}

Together with the input constraint set $\C_{u,i}$ in~\eqref{C3}, the collision avoidance set $\C_{c,i}$ in~\eqref{C1} and sensing range set $\C_{s,i}$ in~\eqref{C2} collectively describe the safety requirements for each pursuer. 

Now we can state the problem studied in this paper.
\begin{problem}\label{problem}
Given the target-pursuit system governed by~\eqref{sys1}, where each pursuer follows its dynamics.
Design the safe control law $u_i$ for all $i\in\mathcal{I}_x$ such that safety sets $\C_{u,i}$ in~\eqref{C3}, $\C_{c,i}$ in~\eqref{C1}, and $\C_{s,i}$ in~\eqref{C2} forward invariant for all $t\geq 0$.
\end{problem}

%% 4. Case study
\section{TRUST-UP Framework Design}
\label{sec-method}
    In this section, we present the design of the TRUST-UP algorithm to solve Problem~\ref{problem}. Figure~\ref{fig:liuchengtu} illustrates the overall architecture of our framework, which integrates model-free RL with CBF-based safety filtering to ensure human-aware safety boundaries. The proposed approach begins by transforming the original system into an augmented system, reformulating the input-constrained problem into an output-constrained problem. To ensure forward invariance of the safety sets $\C_u$, $\C_c$, and $\C_s$ for the pursuers, an adaptive CBF-based safety filter is developed, providing robustness against external disturbances. Finally, the CBF-based safety filter is integrated with a model-free RL policy into the TRUST-UP algorithm. We demonstrate that the TRUST-UP algorithm guarantees feasibility and ensures sensing safety, collision avoidance, and compliance with input constraints for all pursuit UAVs.
\begin{figure}[!ht]
    \centering
    \includegraphics[width=1.05\linewidth]{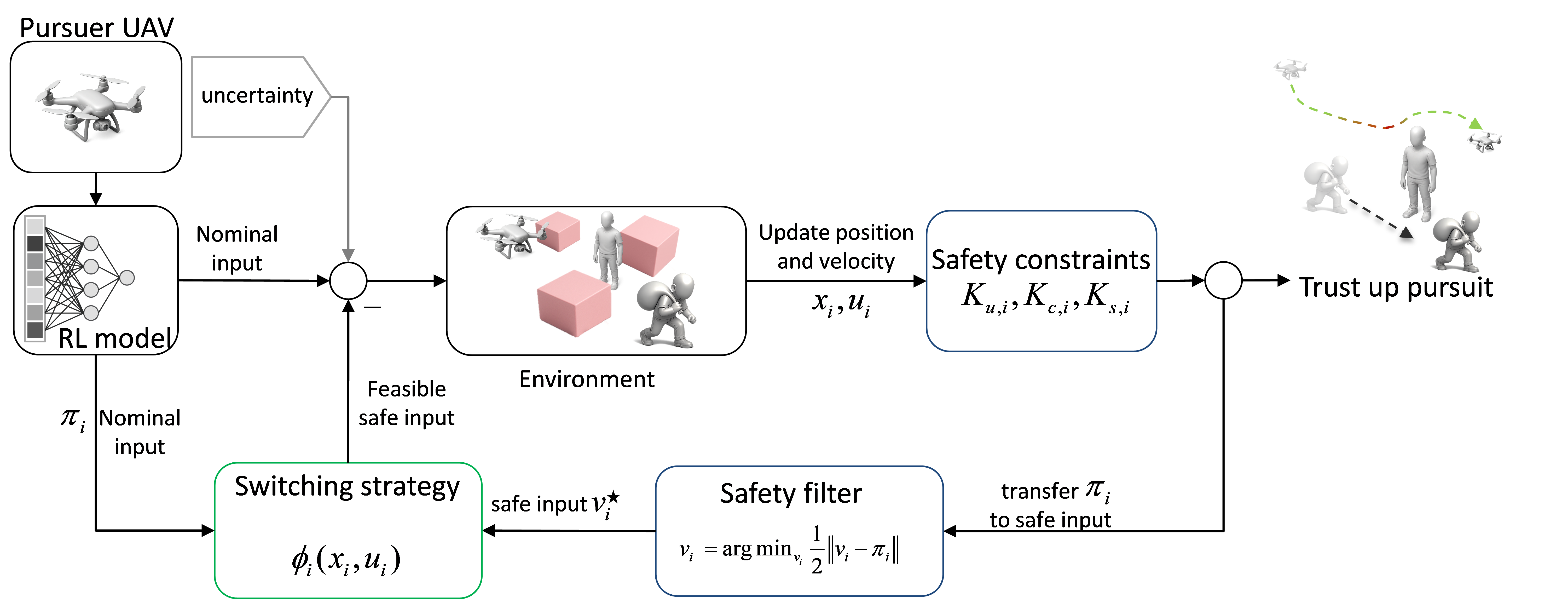}
    \caption{Architecture of the TRUST-UP framework for safe aerial pursuit in human-populated environments. The system integrates a trained RL model that generates nominal control inputs $\pi_i$ with a CBF-based safety filter that enforces human-aware safety constraints. The RL model operates based on the pursuer UAV's state, while the safety filter monitors real-time environmental conditions including human positions with expanded trust radii. Three safety constraints are formulated as a QP problem. The switching strategy improves the practical feasibility of the safety filter
by selecting between the nominal RL action and the filtered command under the
stated admissibility assumptions.}
    \label{fig:liuchengtu}
\end{figure}
\subsection{System transform}\label{subsec-method-1}
% To formally address the control objective stated in Problem~\ref{problem}, we leverage the CBFs framework to ensure both position and input constraints are respected. 
% The objective is to derive explicit conditions for safe target tracking control by formulating appropriate CBFs that maintain the safety-critical properties of the system, including time-varying input limits, collision avoidance, and safe tracking range. 
To provide time-varying bounds on the actual control variable 
$u$, it is natural to place an integrator in the feedback path to augment the system's output as the input of an auxiliary system. Specifically, by introducing an integrator for the control input $u$, the original first-order system in~\eqref{sys1} is transformed into a second-order system, where the time derivative of $u$ is treated as a new auxiliary input $v$. As a result, the system can now be described as:
\begin{equation}\label{sys-aug} \begin{aligned} \dot{x}_i(t) &= f(x_i(t)) + g(x_i(t))u_i(t) + Y(x_i(t)) \theta\\ \dot{u}_i(t) &= v_i(t) + Z(x_i(t))\xi, \end{aligned}\end{equation}
where $v_i \in \mathbb{R}^ n$ is an auxiliary input vector. $\xi\in\R$ is the unknown uncertainty, and $Z\in\R^m$ is a local Lipschitz functions. {This system augmentation in~\eqref{sys-aug} incorporates actuator inertia into the safety constraints, thereby alleviating the actuator overload problem caused by abrupt input changes in traditional first-order CBFs~\cite{deng2025safety,xiao2021high}. In engineering practice, the inherent trade-off between transient response speed and control smoothness introduced by this formulation can be systematically managed by tuning the parameters of the Class $\mathcal{K}$ functions~\cite{xiao2021high,xu2018constrained}.}
\begin{remark}
The uncertainty in the system~\eqref{sys-aug} will always be regarded as sensor faults polluting all the states~\cite{remark-why-d,hung2022image}, which addresses that all the states including $u_i$ are polluted due to sensor faults coinciding in each system state, which is of theoretical and practical significance. 
\end{remark}
% To facilitate our approach, we make the following
% assumption on the structure of the uncertainty in~\eqref{sys-aug}
{To facilitate our approach, we make the following assumption of the uncertainty in~\eqref{sys-aug}, which represents a common engineering practice to handle real-world uncertainties. In practice, even when disturbances such as wind and sensor noise are non-Gaussian, they are often enclosed by bounded convex sets derived from physical and hardware limitations~\cite{ames2020integral,deng2025safety,patil2022adaptive}, so that robust safety constraints remain computationally tractable.}

\begin{assumption}\label{amp-uncertainbound}
The uncertain parameters $\theta$ and $\xi$ belong to the known convex polytope $\theta\in \Theta\subset \R^p$ and $\xi\in \Xi\subset \R^p$, so that there exists the maximum possible estimation error $\varepsilon_\theta, \varepsilon_\xi\in\R^+$ such that $\norm{\theta-\hat \theta}\leq \varepsilon_\theta$ and $\norm{\xi-\hat \xi}\leq \varepsilon_\xi$.
\end{assumption}
{Although the parameter vectors $\theta$ and $\xi$ in~\eqref{sys-aug} are constant, the total lumped disturbances $Y(x_i(t))\theta$ and $Z(x_i(t))\xi$ in~\eqref{sys-aug} remain state-dependent and time-varying through the regressors $Y(x_i(t))$ and $Z(x_i(t))$, thereby providing a tractable approximation of dynamic uncertainties such as urban wind.
}
Then, to deal with the uncertainty in the augmented system~\eqref{sys-aug}, the following Lemma, adapted from~\cite{cohen2022high}, provides a verifiable bound on the estimates $\hat \theta$ and $\hat \xi$, based on the following adaptive law.
{Note that the total combined disturbance in~\eqref{sys1}
functions as a time-varying signal, which adequately represents the dynamic nature of real-world urban wind disturbances.}
\begin{lemma}\label{lem-updatelaw}
Let $t_j\in[t-\Delta T,t]\subset\R$, $j\in\{1,\ldots,M\}$ be the $j$-th sampling time. $M$ represents the total sampling times.  
Given the adaptive law for system~\eqref{sys-aug} as follows
\begin{equation}\label{update law theta}
\begin{multlined}
\dot{\hat \theta}(t)=\gamma_\theta\sum_{j=1}^M\mathcal{Y}_i(t_j)^\top\Big(\Delta x_{i,j}(t_j)-\mathcal{Y}_i(t_j)\hat{\theta}(t_j)\\-\mathcal{F}_i(t_j)-\mathcal{G}_i(t_j)\Big),
\end{multlined}
\end{equation}
\begin{equation}\label{update law xi}
\dot{\hat \xi}(t)=\gamma_\xi\sum_{j=1}^M\mathcal{Z}_i(t_j)^\top\Big(\Delta x_{i,j}(t_j)\!-\!\mathcal{Z}_i(t_j)\hat{\xi}\!-\!\mathcal{V}_i(t_j)\Big),
\end{equation}
where, $\Delta x_{i,j}(t) = x_i(t_j)-x_i(t_j-\Delta T)$,
$\mathcal{F}_i(t_j)=\int_{t_j-\Delta T}^{t_j}f(x_i(\tau))d\tau $, $\mathcal{Y}_i(t_j)=\int_{t_j-\Delta T}^{t_j}Y(x_i(\tau))d\tau $,  $\mathcal{G}_i(t_j) = \int_{t_j-\Delta T}^{t_j}g(x_i(\tau))u_i(\tau)d\tau$, $\mathcal{Z}_i(t_j) = \int_{t_j-\Delta T}^{t_j}Z(x_i(\tau))d\tau$ and $\mathcal{V}_i(t_j) = \int_{t_j-\Delta T}^{t_j}v_i(\tau)d\tau$.
$\gamma_\theta\in\R^+$ and $\gamma_\xi\in\R^+$ are adaptation gains. Provided Assumption~\ref{amp-uncertainbound} holds and $\hat \theta(0)\in\Theta$, $\hat \xi(0)\in\Xi$. then the parameter estimation error $\tilde{\theta}$ and $\tilde{\xi}$ are bounded by two positive constant $\bar \nu\in\R^+$ and $\bar \eta\in\R^+$  
\begin{equation}\label{error bound}
\begin{aligned}
\|\tilde{\theta}(t)\|&\leq\bar \nu\\
\|\tilde{\xi}(t)\|&\leq\bar \eta,
\end{aligned}
\end{equation}
for all $t\geq 0$.
\end{lemma}
Using Assumption~\ref{amp-uncertainbound} and Theorem 2 in~\cite{parikh2019integral} to~\eqref{error bound}, the positive constants $\bar \nu$ and $\bar \eta$ in Lemma~\ref{lem-updatelaw} satisfy
\begin{equation}
\begin{aligned}
\bar\nu &\geq \|\varepsilon_\theta\|e^{-\gamma_\theta\int_0^t\lambda_\theta(\tau)d\tau}\\
\bar\eta &\geq\|{\varepsilon_\xi}\|e^{-\gamma\int_0^t\lambda_\xi(\tau)d\tau},
\end{aligned}
\end{equation}
for all $t\geq 0$, where 
\begin{equation}
\lambda_\theta(t)=E_{\min}\sum_{j=1}^M\mathcal{Y}_i(t_j)^\top\mathcal{Y}_i(t_j),
\end{equation}
\begin{equation}
\lambda_\xi(t)=E_{\min}\sum_{j=1}^M\mathcal{Z}_i(t_j)^\top\mathcal{Z}_i(t_j).
\end{equation}
The above lemma implies that, under the updated law in~\eqref{update law theta} and~\eqref{update law xi},  the parameter estimation error is always bounded by a known value provided the initial parameter estimates are selected such that $\hat \theta(0)\in\Theta$ and $\hat \xi(0) \in \Xi$. These bounded errors enable the development of a robust safety filter to ensure the forward invariance of safety sets $\C_{u,i}$ in~\eqref{C3}, $\C_{c,i}$ in~\eqref{C1}, and $\C_{s,i}$ in~\eqref{C2}, even in the presence of system uncertainties.
\subsection{Safety Filter Design}
{In this subsection, we design auxiliary input $v_i$  to regulate the system in~\eqref{sys-aug}, ensuring that $\C_{u,i}$ in~\eqref{C3}, $\C_{c,i}$ in~\eqref{C1}, and $\C_{s,i}$ in~\eqref{C2} remain forward invariant even in the presence of system uncertainties. To achieve this, we construct three adaptive CBFs for each safety requirement, leveraging the parameter estimation error derived in Lemma~\ref{lem-updatelaw}. Then, we formulate three CBF conditions as a QP problem, whose solution yields the safety filter’s output.}

% In this subsection, we design an adaptive CBF-based safety filter for auxiliary input $v_i$ in system~\eqref{sys-aug} which incorporates the bounded parameter estimation errors derived in Lemma~\ref{lem-updatelaw}, to guarantee the safety of sensing range, collision-avoidance, and input constraint for each pursuer.

Building on this framework, the input constraint for the pursuer is designed to address two key aspects. The first component is a static constraint, which serves as a baseline to filter extreme RL actions, ensuring smooth and safe control inputs. The second component is a dynamic, time-varying constraint that activates during critical scenarios (e.g. $\norm{x_i-q_i}-R_i\to 0$). 
Thus we define the following input-constrained CBF based on the relative positions between the pursuer and the target, denoted by $\zeta_i = x_i - q_i$, as follows
\begin{equation}\label{h3}
h_{u,i}(t) = \kappa(\zeta_i(t))^2 - u_i(t)^2,
\end{equation}
where,
\begin{equation}\label{kappa}
\kappa(\zeta_i(t)) = c + \dfrac{1}{(\zeta_i^\top \zeta_i-\ell^2)^2 + \epsilon},
\end{equation}
with positive constants $c > 0$, $\epsilon > 0$, and $\ell >0$. 
\begin{remark}
The formulation~\eqref{kappa} represents a time-varying input constraint that dynamically adapts based on the relative position of the $i$-th pursuer $x_i$ to its target $q_i$, ensuring that the input signal remains appropriately regulated, thereby guaranteeing safety even under critical scenarios such as target evasive maneuvers or simultaneous obstacle avoidance and tracking.
\end{remark}
Building on the augmented system in Section~\ref{subsec-method-1}, we now proceed to ensures the safety sets $\C_{c,i}$, $\C_{s,i}$ and $\C_{u,i}$ forward invariant for all $x_i\in\mathcal{X}$, $i\in\mathcal{I}_x$. Specifically, we design the auxiliary controller $v_i$ for system~\eqref{sys-aug}, such that CBFs defined by~\eqref{h3},~\eqref{h1} and~\eqref{h2} satisfy
$h_{c,i}(t)\geq 0$, $h_{s,i}(t)\geq 0$, and $h_{u,i}(t)\geq 0$ for all $i\in\mathcal{I}_x$ and $t\geq 0$. 
For simplicity of notations, we denote $f_i=f(x_i)$, $g_i=g(x_i)$, $Y_i=Y(x_i)$, $Z_i=Z(x_i)$, and omit the time variable~$t$ in the rest of the paper. 

The following lemma provides conditions under which the input constraint safety set $\mathcal{C}_{u,i}$ forward invariant.
% With the adaptive parameter estimation laws established in Lemma~\ref{lem-updatelaw}, we first proceed to guarantee the safety set of input constraints $\C_u$. The following lemma provides conditions under which the input constraint safety set $\mathcal{C}_u$ forward invariant.
\begin{lemma}\label{lem-kcbf3}
Suppose Assumption~\ref{asm-p dot p} and Assumption~\ref{amp-uncertainbound} holds for all $i\in\mathcal{I}_x$ and $t\geq 0$. Let the parameter estimation error be bounded as in Lemma~\ref{lem-updatelaw}, and define the set of admissible control inputs for the $i$-th pursuer as:
\begin{equation}\label{kcbf3}
\begin{aligned}
K_{u,i}(x_i,u_i,\hat \theta,\hat \xi)
=\sup_{v_i\in\R^m}\Big\{L_{u,i} -u_i^\top v_i + \alpha(h_{u,i})\geq 0
\Big\},
\end{aligned}
\end{equation}
where
\begin{equation}\label{Lcbf3}
\begin{aligned}
&L_{u,i}(x_i,u_i,\hat \theta,\hat \xi)\\
&\begin{multlined}[t]
=\dfrac{-2\kappa (\zeta_i^\top\zeta_i\!-\ell^2)}{(\zeta_i^\top\zeta_i\!-\!\ell^2)^2+\epsilon)^2}\zeta_i^\top (f_i\!+\!g_iu_i \!+\! Y_i\hat \theta) \!-\!u_i^\top Z_i\hat \xi
\\
-\norm{\dfrac{2\kappa (\zeta_i^\top\zeta_i-\ell^2)}{(\zeta_i^\top\zeta_i-\ell^2)^2+\epsilon)^2}\zeta_i}\left(\norm{Y_i}\nu+\rho_v\right) \\  -\norm{u_i^\top Z_i}\eta.
\end{multlined}
\end{aligned}
\end{equation}
Then, any Lipschitz continuous controller $v_i\in K_{u,i}(x_i,\hat \theta,\hat \xi)$ guarantees the forward invariance of set $\C_{u,i}$ in~\eqref{C3} regarding to system~\eqref{sys-aug} for all $x_i\in\R^n$, $\theta\in\Theta$, $\xi\in\Xi$ and $t\geq 0$.
\end{lemma}
The proof of this lemma~\ref{lem-kcbf3} is provided in the Appendix.

With the safety of the control input ensured by Lemma~\ref{lem-kcbf3}, we now extend our analysis to derive safe conditions for collision avoidance. Given a position of $i$-th pursuer $x_i$, the objective is to ensure
\begin{equation}
\varkappa_k(t)  = \norm{x_i(t) - p_k(t)} \geq  r_i,
\end{equation}
for all $k\in \mathcal{I}_{k}$ and all $t\geq 0$.
The following lemma provides sufficient conditions of forward invariant of the safety set $\mathcal{C}_{c, i}$, which indicates the $i$-th pursuer is collision-free.
\begin{lemma}\label{lem-kcbf1}
Let $\iota_k\in\R^+$ be a positive constant for all $k\in\mathcal{I}_k$.
Suppose Assumption~\ref{asm-p dot p} and Assumption~\ref{amp-uncertainbound} holds. Let the parameter estimation error be bounded as in Lemma~\ref{lem-updatelaw}, and define the set of admissible control inputs for the $i$-th pursuer as:
\begin{equation}
\begin{aligned}
&K_{c,i,k}(x_i,\hat \theta,\hat \xi)\\
&\begin{multlined}[t]
=\sup_{v_i\in\R^m}2\Bigg\{L_{c,i,k}
- \varkappa_k^\top g_iv_i
 + \alpha_k(h_{i,k}(x_i))\geq 0
\Bigg\},
\end{multlined}
\end{aligned}
\end{equation}
where
\begin{equation}
\begin{aligned}
&L_{c,i,k}(x_i,\hat \theta,\hat \xi)\\
&\begin{multlined}[t]
=\varkappa_k^\top(\dot f_i + \dot g_iu_i) + \varkappa_k^\top\dot Y_i\hat \theta-\norm{\varkappa_k^\top\dot Y_i} \nu - \norm{\varkappa_k}\rho_a \\ + \varkappa_k^\top g_iZ_i\hat\xi -\norm{\varkappa_k^\top g_iZ_i}\eta\\
+ \iota_k \Big( \varkappa_k^\top(f_i + g_iu_i) + \varkappa_k^\top Y_i\hat \theta \\- \norm{\varkappa_k^\top Y_i}\nu\Big).
\end{multlined}
\end{aligned}
\end{equation}
Then any Lipschitz continuous controller $v_i\in  K_{c,i} = \bigcap_{k\in\mathcal{I}_k} K_{c,i,k}$
renders the safety of $\C_{c,i}$ forward invariant for all $k\in\mathcal{I}_k$, $\theta\in\Theta$, $\xi\in\Xi$ and all $t\geq 0$.
\end{lemma}
The proof of this lemma~\ref{lem-kcbf1} is provided in the Appendix.

Subsequently, the conditions to ensure the safety of the sensing range $\mathcal{C}_{s,i}$ in~\eqref{C2} are derived based on the relative position $\zeta_i=x_i-q_i$ between the pursuer and the target, as stated in Lemma~\ref{lem-kcbf2}.
\begin{lemma}\label{lem-kcbf2}
Let $\imath_t\in\R^+$ be a positive constant.
Suppose Assumption~\ref{asm-p dot p} and Assumption~\ref{amp-uncertainbound} holds for all $t\geq 0$. Let the parameter estimation error be bounded as in Lemma~\ref{lem-updatelaw}, and define the set of admissible control inputs for the $i$-th pursuer as:
\begin{equation}\label{Kcbf s}
\begin{aligned}
&K_{s,i}(x_i,\hat \theta,\hat \xi)\\
&\begin{multlined}[t]
=\sup_{v_i\in\R^m}2\Bigg\{L_{s,i}
 + \zeta_i^\top g_iv_i -\alpha_t(h_{s,i}(x_i))\geq 0\Bigg\}.
\end{multlined}
\end{aligned}
\end{equation}
where
\begin{equation}\label{Lcbf s}
\begin{aligned}
&L_{s,i}(\zeta_i,\hat \theta,\hat \xi)\\
&\begin{multlined}[t]
=
-\imath_t\zeta_i^\top (f_i + g_iu_i) - \zeta_i^\top (\imath_t Y_i+\dot Y_i)\hat \theta +\norm{\zeta_i}\rho_a
\\
\!-\!\zeta_i^\top(\dot f_i \!+\! \dot g_iu_i)\!-\! \zeta_i^\top g_iZ_i\hat\xi  \!+\! \norm{\zeta_i^\top g_i Z_i}\eta\\
\!+\!\left(\norm{\imath_t\zeta_i^\top Y}\!+\!\norm{\zeta_i^\top \dot Y_i}\right)\nu.
\end{multlined}
\end{aligned}
\end{equation}
Then any Lipschitz continuous controller $v_i\in K_{s,i}$ renders the safety of $\C_{c,i}$ forward invariant for all $x_i\in\R^n$, $\theta\in\Theta$, $\xi\in\Xi$ and all $t\geq 0$.
\end{lemma}
Since the proof of Lemma~\ref{lem-kcbf2} is analogous to that of Lemma~\ref{lem-kcbf1}, it is omitted here.
The above Lemmas~\ref{lem-kcbf3}, \ref{lem-kcbf1}, and \ref{lem-kcbf2} collectively establish that the proposed safety constraints—collision avoidance, sensing range, and input constraint—are forward invariant under appropriate control input designs for each pursuer in the target-pursuit system~\eqref{sys-aug}. Building on these results, we construct a CBF-QP formulation that integrates the trained RL action $\pi_i(t)$ as the nominal control input for pursuer $i$ at time $t$. This formulation ensures compliance with the safety constraints established in Lemmas~\ref{lem-kcbf3}, \ref{lem-kcbf1}, and \ref{lem-kcbf2}, yielding the safety adaptive controller $v_i^\star$, as follows:

\textbf{CBF-QP Problem for Pursuer $i$}
\begin{equation}\label{QP}
\begin{aligned}
&v_i^\star(t)=\mathop{\arg\min}_{v_i\in\mathbb{R}^{m}}\dfrac{1}{2}\left\|v_i(t)-\pi_{i}(t)\right\|^2 \\
&\text{s.t. }\\
&v_i\!\in\! K_{u,i}(x_i,u_i,\hat \theta,\hat \xi)\!\bigcap \!K_{c,i}(x_i,\hat \theta,\hat \xi)\!\bigcap\! K_{s,i}(x_i,\hat \theta,\hat \xi).
\end{aligned}
\end{equation}

Given the above CBF-QP problem, it is essential to investigate the feasibility of the proposed safety filter. Specifically, we aim to integrate the RL control input $\pi_i$ with the safety filter into a unified algorithm, ensuring that when the RL input violates safety conditions, the safety filter is activated and remains feasible for each safety requirement.

\subsection{TRUST-UP Algorithm}
In this part, we proposed the overall framework of our TRUST-UP algorithm for the pursuer control problem, which integrates model-free RL with a safety filter to ensure the satisfaction of the safety constraints $\C_{u,i}$ in~\eqref{C3}, $\C_{c,i}$ in~\eqref{C1} and $\C_{s,i}$ in~\eqref{C2}. 
Specifically, there is a switch strategy in our TRUST-UP algorithm that switches between the RL action $\pi_i$ and the safety filter in~\eqref{QP}, depending on whether the safety filter is activated. To formalize the feasibility of this framework, we present the main theorem of this paper, which proves the switch strategy ensures the CBF-QP problem satisfies the safety conditions in Lemmas~\ref{lem-kcbf3}, \ref{lem-kcbf1}, and \ref{lem-kcbf2}, whenever the safety filter is activated. Consequently, the proposed TRUST-UP algorithm guarantees that the pursuer control problem stated in Problem~\ref{problem} is fully satisfied.

% Let $\pi_i$ denote the trained RL action for the $i$-th pursuer. The following CBF-QP formulation is constructed by incorporating RL actions $pi_i$ with safety constraint in Lemma~\ref{lem-kcbf3}, Lemma~\ref{lem-kcbf1} and Lemma~\ref{lem-kcbf2}, to obtain the safe adaptive controller $v_i^\star$, as follows:

% \textbf{CBF-QP Problem for Pursuer $i$}
% \begin{equation}\label{QP}
% \begin{aligned}
% v_i^\star&=\mathop{\arg\min}_{v_i\in\mathbb{R}^{m}}\dfrac{1}{2}\left\|v_i-\pi_{i}\right\|^2 \\
% &\text{s.t. }\\
% &v_i\in K_{u,i}(x_i,\hat \theta,\hat \xi)\bigcap K_{c,i}(x_i,\hat \theta,\hat \xi)\bigcap K_{s,i}(x_i,\hat \theta,\hat \xi),
% \end{aligned}
% \end{equation}
Given the above CBF-QP problem in~\eqref{QP}, it is necessary to develop a feasible continuous control law to satisfy each safety condition. To achieve this, we introduce a switch strategy $\pi_i$ that activates the safety filter in~\eqref{QP} by switching from the nominal RL action $\pi_i$ whenever a CBF constraint is violated. For this purpose, we define the following regions:
\begin{equation}\label{asm3need}
\begin{aligned}
&\Pi_{u,i} \!=\! \Big\{u_i\in \C_{u,i}\mid L_{u,i} + \alpha h_{u,i} - u_i^\top \pi_i\geq 0 \Big\},\\
&\Pi_{s,i} \!=\! \Big\{x_i\in \C_{s,i}\mid L_{s,i} \!+\! \alpha h_{s,i} \!-\! \zeta_i^\top g_i \pi_i\geq 0 \Big\},\\
&
\Pi_{c,i} \!=\! \Big\{\!x_i\in \C_{c,i}\!\mid\!\bigcap_{k\in\mathcal{I}_k} \Big(L_{c,i,k} + \alpha h_{c,i,k} \\ 
&\qquad\qquad\qquad\qquad \qquad\qquad\quad- \varkappa_k^\top g_i \pi_i\geq 0 \Big)\Big\},
\end{aligned}
\end{equation}
which represents the case when nominal RL action $\pi_i$ is safe for rendering the safety sets forward invariant. Then we define
\begin{equation}
\begin{multlined}
\mathcal{R}_{1} \!=\! \Big\{(x_i,u_i) \mid (u_i\in \Pi_{u,i}) \cap \Big(x_i\in (\Pi_{c,i}\cap\Pi_{s,i}) \Big)\Big\},
\end{multlined}
\end{equation}
and 
\begin{equation}
    \mathcal{R}_2 = \lnot \mathcal{R}_{1}.
\end{equation}
Obviously, when $(x_i, u_i)\in \mathcal{R}_{1}$, the safety filter is inactive since $v_i=\pi_i$, and  when $(x_i, u_i)\in \mathcal{R}_{2}$, we solved the safe control $v_i^\star$ by QP problem in~\eqref{QP}. Therefore, we proposed the following hybrid control law:
\begin{equation}\label{hybrid}
\phi_i(x_i,u_i) =     \begin{cases} 
        \pi_i(x_i,u_i) & \text{if } (x_i, u_i)\in\mathcal{R}_1, \\ 
        v_i^\star (x_i,u_i)  & \text{if } (x_i, u_i)\in\mathcal{R}_2. 
    \end{cases}
\end{equation}
{To make the feasibility statement of~\eqref{QP} under~\eqref{asm3need}, we need the following assumption.}
{
\begin{assumption}
For each pursuer $i$, and for all states in the considered operating domain, the joint admissible set
\[
\mathcal{F}_i(x_i,u_i)
:=
K_{u,i}(x_i,u_i,\hat{\theta},\hat{\xi})
\cap
K_{c,i}(x_i,\hat{\theta},\hat{\xi})
\cap
K_{s,i}(x_i,\hat{\theta},\hat{\xi})
\]
is non-empty.
\end{assumption}
}
Now we can show the main Theorem for the feasibility of~\eqref{QP}.
\begin{theorem}\label{thm-feasible}
If the initial conditions of~\eqref{h1},~\eqref{h2} and~\eqref{h3} hold for 
\begin{equation} \hbar_{c,i,k}(0)>0, \quad \hbar_{s,i}(0)>0, \quad h_{u,i}(0)>0, 
\end{equation}
and given the hybrid control law $\phi_i(x_i,u_i)$ in~\eqref{hybrid}. 
Then the QP problem in~\eqref{QP} is feasible and has a unique solution of  $x_i\in\left(\bigcap_{k\in\mathcal{I}_k}\C_{c,i,k}\right)\cap \C_{s,i}$ and $u_i\in\C_{u,i}$ for all $i\in\mathcal{I}_x$ and $t\geq 0$.
\end{theorem}
\begin{proof}
To determine the feasibility of maintaining these constraints
simultaneously, we construct the Lagrangian function:
\begin{equation}
\begin{aligned}
&L_i ( v_i, \lambda_{i}, \lambda_{s,i}, \lambda_{u,i} )\\
&\begin{multlined}[t]
=\dfrac{1}{2}\| v_i-\pi_i \|^{2}-\sum_{k\in\mathcal{I}_k}\lambda_{i,k} \hbar_{c,i,k}-\lambda_{s,i} \hbar_{s,i}-\lambda_{u,i} \dot{h}_{u,i},
\end{multlined}
\end{aligned}
\end{equation}
where $\lambda_{i,k}, \lambda_{s,i}$, and $\lambda_{u,i}$ should always satisfy dual feasibility s.t. $\lambda_{i,k}, \lambda_{s,i}, \lambda_{u,i} \geq 0$ , $i\in\mathcal{I}$, are the Lagrange multipliers associated with each constraint.
According to the KKT conditions, the solutions of the QP program are optimal and unique if the following stationary condition
\begin{equation}\label{KKT-1}
\begin{aligned}
&\frac{\partial L_i}{\partial v_i}\\
&\begin{multlined}[t]
=(v_i \!-\! \pi_i) \!-\! \sum_{k\in\mathcal{I}_k} \lambda_{c,i,k} \frac{\partial \hbar_{c,i,k}}{\partial v_i} \!-\! \lambda_{s,i} \frac{\partial \hbar_{s,i}}{\partial v_i} \!-\! \lambda_{u,i} \frac{\partial \dot h_{u,i}}{\partial v_i}
\end{multlined}\\
&=0,
\end{aligned}
\end{equation}
and complementary slackness
\begin{equation}\label{KKT-complementary slackness}
\begin{aligned}
\lambda_{s,i} (L_{s,i} + \alpha h_{s,i} + \zeta_i^\top g_iv_i) \!&=\!0 , \\
\lambda_{u,i} (L_{u,i} + \alpha h_{u,i} + u_i^\top v_i) \!&=\!0 \\
\sum_{k\in\mathcal{I}_k}\!\lambda_{c,i,k}(L_{c,i,k} \!+\! \alpha h_{u,i}\!+\! \varkappa_k^\top g_iv_i) \!&=\! 0, \\
\end{aligned}
\end{equation}
hold with respect to $v_i$.
Additionally, we define the active set $\mathcal{A}_j$, $j=\{1,2,3\}$ for the QP problem in~\eqref{QP} as the set of constraints that are active at a given solution $v_i^\star$, i.e., the constraints for which the equality $\hbar_{c,i,k}=0$, $\hbar_{s,i}=0$, or ${\dot h}_{u,i}=0$ holds. For all inactive constraints, the corresponding inequality holds $\hbar_{c,i,k}>0$, $\hbar_{s,i}>0$, and ${\dot h}_{u,i}>0$.

To analyze the feasibility of the QP problem, we individually consider each constraint being active while others are inactive. Specifically,  there are three cases to be considered:

\textit{Case 1: $\mathcal{A}_1 = \{k\in\mathcal{I}_k\mid \hbar_{c,i,k}(v_i)=0\}$.}

For $\varkappa_k\in\mathcal{R}_2$ the KKT conditions result in
\begin{subequations}\label{kkt-1}
    \begin{align}
        v_i-\pi_i-\sum_{k\in\mathcal{I}_k}\lambda_{i,k} \varkappa_k^\top g_i &= 0 \label{kkt-1a} \\
        \sum_{k\in\mathcal{I}_k}\lambda_{c,i,k}(L_{c,i,k} + \alpha h_{u,i}+ \varkappa_k^\top g_iv_i) &= 0 \label{kkt-1b} \\
        \lambda_{c,i,k} &\geq 0 \label{kkt-1c}
    \end{align}
\end{subequations}
for all $k\in\mathcal{I}_k$. Using hybrid control law in~\eqref{hybrid}, we have 
\begin{equation}\label{lghpi-1}
L_{c,i,k} + \alpha h_{c,i,k}(x_i) - \varkappa_k^\top g_i \pi_i\leq 0 
\end{equation}
Apply~\eqref{kkt-1b} to~\eqref{kkt-1a}, we obtains
\begin{equation}\label{opt-1}
v_i = -\dfrac{L_{c,i,k} + \alpha h_{c,i,k}(x_i)}{\norm{\varkappa_k^\top g_i}^2} g_i^\top\varkappa_k,
\end{equation}
and the stationarity conditions of~\eqref{kkt-1a}
\begin{equation}\label{temp1}
-\dfrac{L_{c,i,k} + \alpha h_{c,i,k}(x_i)}{\norm{\varkappa_k^\top g_i}^2} g_i^\top\varkappa_k -\pi_i-\sum_{k\in\mathcal{I}_k}\lambda_{i,k} \varkappa_k^\top g_i = 0.
\end{equation}
Using~\eqref{lghpi-1} to~\eqref{temp1}, the following inequality is obtained
\begin{equation}
\begin{multlined}
\sum_{k\in\mathcal{I}_k}\lambda_{i,k} = -\Bigg( 
\dfrac{L_{c,i,k} + \alpha h_{c,i,k}(x_i)}{\norm{\varkappa_k^\top g_i}^2} g_i^\top\varkappa_k + \pi_i \Bigg)\dfrac{g_i^\top\varkappa_k}{\norm{\varkappa_k^\top g_i}^2}
\end{multlined}
\end{equation}

\begin{equation}\begin{aligned}
&\sum_{k\in\mathcal{I}_k}\lambda_{i,k} \\
&\begin{multlined}[.85\linewidth]
= -\Bigg( 
\dfrac{L_{c,i,k} + \alpha h_{c,i,k}(x_i)}{\norm{\varkappa_k^\top g_i}^2} g_i^\top\varkappa_k + \pi_i \Bigg)\dfrac{g_i^\top\varkappa_k}{\norm{\varkappa_k^\top g_i}^2}
\end{multlined}
\\
&\begin{multlined}[.85\linewidth]
=
-\dfrac{L_{c,i,k} + \alpha h_{c,i,k}(x_i)}{\norm{\varkappa_k^\top g_i}^2}+ \pi_i\dfrac{g_i^\top\varkappa_k}{\norm{\varkappa_k^\top g_i}^2}
\end{multlined}
\\
&\begin{multlined}[.85\linewidth]
\geq
-\dfrac{L_{c,i,k} + \alpha h_{c,i,k}(x_i)}{\norm{\varkappa_k^\top g_i}^2} + \dfrac{L_{c,i,k} + \alpha h_{c,i,k}(x_i)}{\norm{\varkappa_k^\top g_i}^2} 
\end{multlined}\\
&\geq 0,
\end{aligned}\end{equation}
which implies the active set $\mathcal{A}_1 = \{k\in\mathcal{I}_k\mid \hbar_{c,i,k}(v_i)=0\}$  is valid for all $k\in\mathcal{I}_k$. Since
the KKT conditions are met and the unique solution for the QP problem in~\eqref{QP} is given by~\eqref{opt-1}, the QP problem
is feasible for all $x_i\in\mathcal{R}_2$ under the validated active set $\mathcal{A}_1$.

\textit{Case 2: $\mathcal{A}_2 = \{\hbar_{s,i}(v_i)=0\}$.}

For $\zeta_i\in\mathcal{R}_2$ the KKT conditions result in
\begin{subequations}\label{kkt-2}
    \begin{align}
        v_i-\pi_i-\lambda_{s,i} \zeta_i^\top g_i &= 0, \label{kkt-2a} \\
        \lambda_{s,i}(L_{s,i} + \alpha h_{s,i}+ \zeta_i^\top g_iv_i) &= 0, \label{kkt-2b} \\
        \lambda_{s,i} &\geq 0. \label{kkt-2c}
    \end{align}
\end{subequations}
Using hybrid control law in~\eqref{hybrid}, we have 
\begin{equation}\label{lghpi-2}
L_{s,i} + \alpha h_{s,i}(x_i) - \zeta_i^\top g_i \pi_i\leq 0 
\end{equation}
Apply~\eqref{kkt-2b} to~\eqref{kkt-2a}, we obtains
\begin{equation}\label{opt-2}
v_i = -\dfrac{L_{s,i} + \alpha h_{s,i,k}(x_i)}{\norm{\zeta_i^\top g_i}^2} g_i^\top\zeta_i,
\end{equation}
Using~\eqref{lghpi-2} to~\eqref{kkt-2a}, the following inequality is obtained
\begin{equation}\begin{aligned}
\lambda_{s,i} 
&= -\Bigg( 
\dfrac{L_{s,i} + \alpha h_{s,i}(x_i)}{\norm{\zeta_i^\top g_i}^2} g_i^\top\zeta_i + \pi_i \Bigg)\dfrac{g_i^\top\zeta_i}{\norm{\zeta_i^\top g_i}^2}
\\
&\geq 0,
\end{aligned}\end{equation}
which implies the active set $\mathcal{A}_i = \{\hbar_{s,i}(v_i)=0\}$  is valid for all $\zeta_i\in\mathcal{R}_2$. Since
the KKT conditions are met and the unique solution for the QP problem in~\eqref{QP} is given by~\eqref{opt-2}, the QP problem
is feasible for all $x_i\in\mathcal{R}_2$ under the validated active set $\mathcal{A}_2$.

\textit{Case 3: $\mathcal{A}_3 = \{\hbar_{u,i}(v_i)=0\}$.}

For $u_i\in\mathcal{R}_2$ the KKT conditions result in
\begin{subequations}\label{kkt-3}
    \begin{align}
        v_i-\pi_i-\lambda_{u,i} u_i &= 0, \label{kkt-3a} \\
        \lambda_{u,i}(L_{u,i} + \alpha h_{u,i}+ u_i^\top v_i) &= 0, \label{kkt-3b} \\
        \lambda_{u,i} &\geq 0. \label{kkt-3c}
    \end{align}
\end{subequations}
Using hybrid control law in~\eqref{hybrid}, we have 
\begin{equation}\label{lghpi-3}
L_{u,i} + \alpha h_{u,i}(x_i) - u_i^\top \pi_i\leq 0 
\end{equation}
Apply~\eqref{kkt-3b} to~\eqref{kkt-3a}, we obtains
\begin{equation}\label{opt-3}
v_i = -\dfrac{L_{u,i} + \alpha h_{u,i}(x_i)}{\norm{u_i}^2} u_i^\top,
\end{equation}
Using~\eqref{lghpi-3} to~\eqref{kkt-3a}, the following inequality is obtained
\begin{equation}\begin{aligned}\label{kkt-lambda2}
\lambda_{u,i} 
&= -\Bigg( 
\dfrac{L_{u,i} + \alpha h_{u,i}(x_i)}{\norm{u_i}^2} u_i^\top + \pi_i \Bigg)\dfrac{u_i^\top}{\norm{u_i}^2}
\\
&\geq 0,
\end{aligned}\end{equation}
which implies the active set $\mathcal{A}_i = \{\hbar_{u,i}(v_i)=0\}$  is valid for all $u_i\in\mathcal{R}_2$. Since
the KKT conditions are met and the unique solution for the QP problem in~\eqref{QP} is given by~\eqref{opt-2}, the QP problem
is feasible for all $u_i\in\mathcal{R}_2$ under the validated active set $\mathcal{A}_3$.

Lastly, we prove the Lipschitz continuity of $v_i^\star$ for all $(x_i,u_i)\in\mathcal{R}_2$. The coefficient matrix of the CBF-QP problem in~\eqref{QP} is given by 
\begin{equation}
V(x_i,u_i) = 
\begin{bmatrix}
 u_i^\top\\
 -\zeta_i^\top g_i\\
    -\sum_{k\in\mathcal{I}_k}\varkappa_k^\top g_i
\end{bmatrix}.
\end{equation}
Following Theorem 3.1 of~\cite{hager1979lipschitz} and given that the solution in~\eqref{opt-1},~\eqref{opt-2} and~\eqref{opt-3} are well defined as $\varkappa_k^\top g_i$, $\zeta_i^\top g_i$ and $u_i$ is bounded away from $0$ in any bounded subset of $\mathcal{R}_2$ as $h_{c,i,k}$, $h_{s,i}$ (relative degree $2$), and $h_{u,i}$ are CBFs for system~\eqref{sys-aug}. Therefore, the matrix $V(x_i,u_i)$ is full rank, ensureing $v_i^\star$ is Lipschitz continuous for all $(x_i,u_i)\in\mathcal{R}_2$.
Since the CBF-QP problem in~\eqref{QP} is utilized as part of the hybrid control law in~\eqref{hybrid}, and thus it is only employed for $(x_i,u_i)\in\mathcal{R}_2$ to compute a control input. For $(x_i,u_i)\in\mathcal{R}_1$, the nominal RL action is safe and the CBF constraints are inactive. Therefore, using Corollary 3.3 of~\cite{van2024unifying}, the hybrid control law $\phi(x_i,u_i)$ in~\eqref{hybrid}  is continuous for all $x_i\in\left(\bigcap_{k\in\mathcal{I}_k}\C_{c,i,k}\right)\cap \C_{s,i}$ and $u_i\in\C_{u,i}$ for any $i\in\mathcal{I}_x$ and $t\geq 0$ thus completing the proof.
\end{proof}
% {\begin{remark}\label{remark:feasibility}
% Theorem~\ref{thm-feasible} proves that all individual constraints satisfy the non-vanishing control gradient condition. If an extreme scenario occurs where multiple constraints are activated strictly simultaneously, our framework can naturally adopt a smooth Boolean composition~\cite{usevitch2022adversarial, cohen2024safety,glotfelter2020nonsmooth} to aggregate them into a single global constraint:
% \begin{equation}\label{eq:softmax}
%  h_{soft, i}(x_i, u_i) = -\frac{1}{\rho} \ln \left( \sum_{k \in \mathcal{K}_i} e^{-\rho h_{c,i,k}} + e^{-\rho h_{s,i}} + e^{-\rho h_{u,i}} \right) \ge 0,
% \end{equation}
% where $\rho > 0$ is a smoothing parameter. By reducing multiple constraints into a single active boundary, this formulation inherently ensures a solvable QP even under simultaneous activation.
% \end{remark}}

Building on Theorem~\ref{thm-feasible}, which establishes the feasibility of the hybrid control law in satisfying the safety constraints~\eqref{C1},~\eqref{C2} and~\eqref{C3},  we now present the framework of our TRUST-UP algorithm for our pursuit control problem. The TRUST-UP algorithm integrates the safety filter with model-free RL, ensuring that each pursuer can safely pursue its target with the sensing limitations and input constraints while avoiding collisions with obstacles and other UAVs. The complete framework is outlined in Algorithm~\ref{alg:safety_rl_execution}.
\begin{algorithm}[t]
\caption{TRUST-UP for UAV Target Pursuit}
\label{alg:safety_rl_execution}
\SetAlgoLined
\KwIn{Environment dynamics $ f $, $ g $; 

Safety constraints $ \mathcal{C}_u $, $ \mathcal{C}_t $, $ \mathcal{C}_c $; 

$m$ Static obstacles $\mathcal{O}=\{o_1,\ldots,o_m\}$; 

$n$ Pursuer $\mathcal{X}=\{x_1,\ldots,x_n\}$; 

$n$ Target $\mathcal{Q}=\{q_1,\ldots,q_n\}$; 

$n$ Trained RL policy $\pi_i, i=\{1,\ldots,n\}$;

Environment disturbances $\theta$, $\xi$.}
\KwOut{$x_i\in\left(\bigcap_{k\in\mathcal{I}_k}\C_{c,i,k}\right)\cap \C_{s,i}$ and $u_i\in\C_{u,i}$.}
\begin{enumerate}
    \item Initialize $n$ pursuers $\mathcal{X}(0) = \{x_1(0),\ldots, x_n(0)\}$, $n$ targets $\mathcal{Q}(0) = \{q_1(0)\ldots,q_n(0)\}$, and $m$ static obstacles $\mathcal{O}= \{o_1(0)\ldots,o_m(0)\}$.
    \item Set initial parameters for each pursuer $x_i(0)$, $u_i(0)$, $f_i(0)$, $g_i(0)$, $Y_i(0)$, $Z_i(0)$, $\hat{\theta}(0)$, $\hat{\xi}(0)$, $i\in\mathcal{I}_x$.
    \item For each time step $t$:

        \begin{enumerate}
            \item Get $\mathcal{P}_{-i}=\{p_1\ldots p_{2n+m-1}\}=(\mathcal{X}\backslash x_i)\cup \mathcal{Q}\cup \mathcal{O}$,
            \item Compute relative positions:
            \begin{itemize}
                \item Target distance $\zeta_i = x_i - q_i$,
                \item Relative positions to other targets, pursuers, and obstacles $\varkappa_k = x_i-p_k$.
            \end{itemize}
            \item Evaluate hybrid control law $\phi(x_i, u_i)$ using~\eqref{hybrid}.
            \item \textbf{If} safety filter is activated:
            \begin{itemize}
                \item Compute auxiliary control input $v_i^\star$ by solving the CBF-QP problem in~\eqref{QP}.
                \item Auxiliary control input $v_i = v_i^\star$.
            \end{itemize}
            \item \textbf{Else}: Auxiliary control input $v_i = \pi_i(x_i)$.
            \item Compute $x_i$ and $u_i$ from $v_i$ by dynamics~\eqref{sys-aug}.
        \end{enumerate}
    \end{enumerate}
\end{algorithm}

%% 5. Conclusion
\section{Simulation Results}
\label{sec-experiment}
    \begin{figure}
    \centering
    \includegraphics[width=0.6\linewidth]{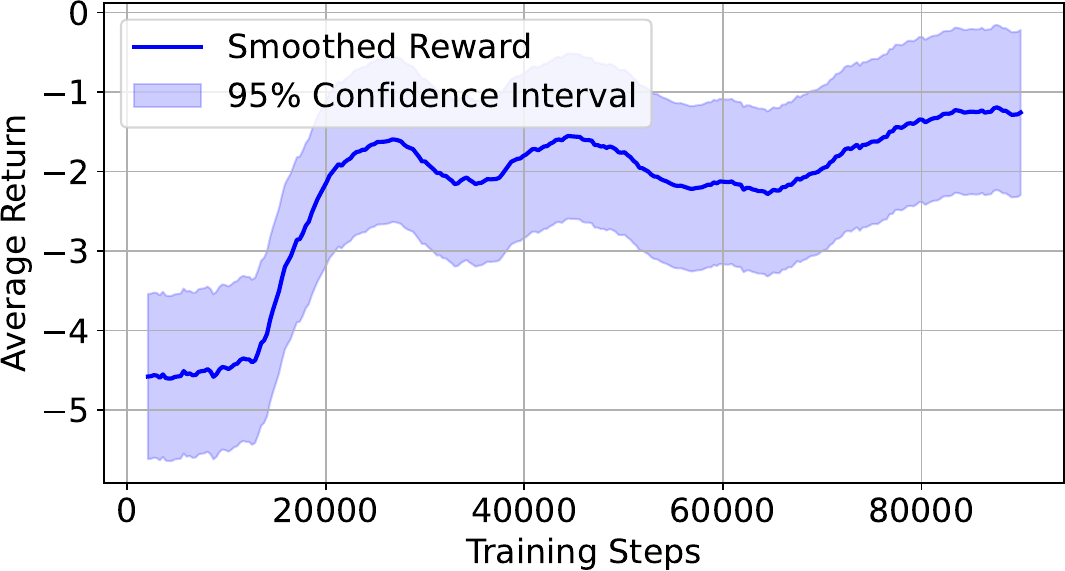}
    \caption{The average return training curves of SAC  by running $3$ times with different seeds. The lines and shaded areas represent
the average return and the 95\% confidence interval, respectively. }
    \label{fig_training}
\end{figure}
\begin{figure*}[!htb]
    \centering
     \hspace*{0.4cm}
    % 第二行的2张图片
    \begin{minipage}{0.3\textwidth}
        \centering
        \includegraphics[width=\linewidth]{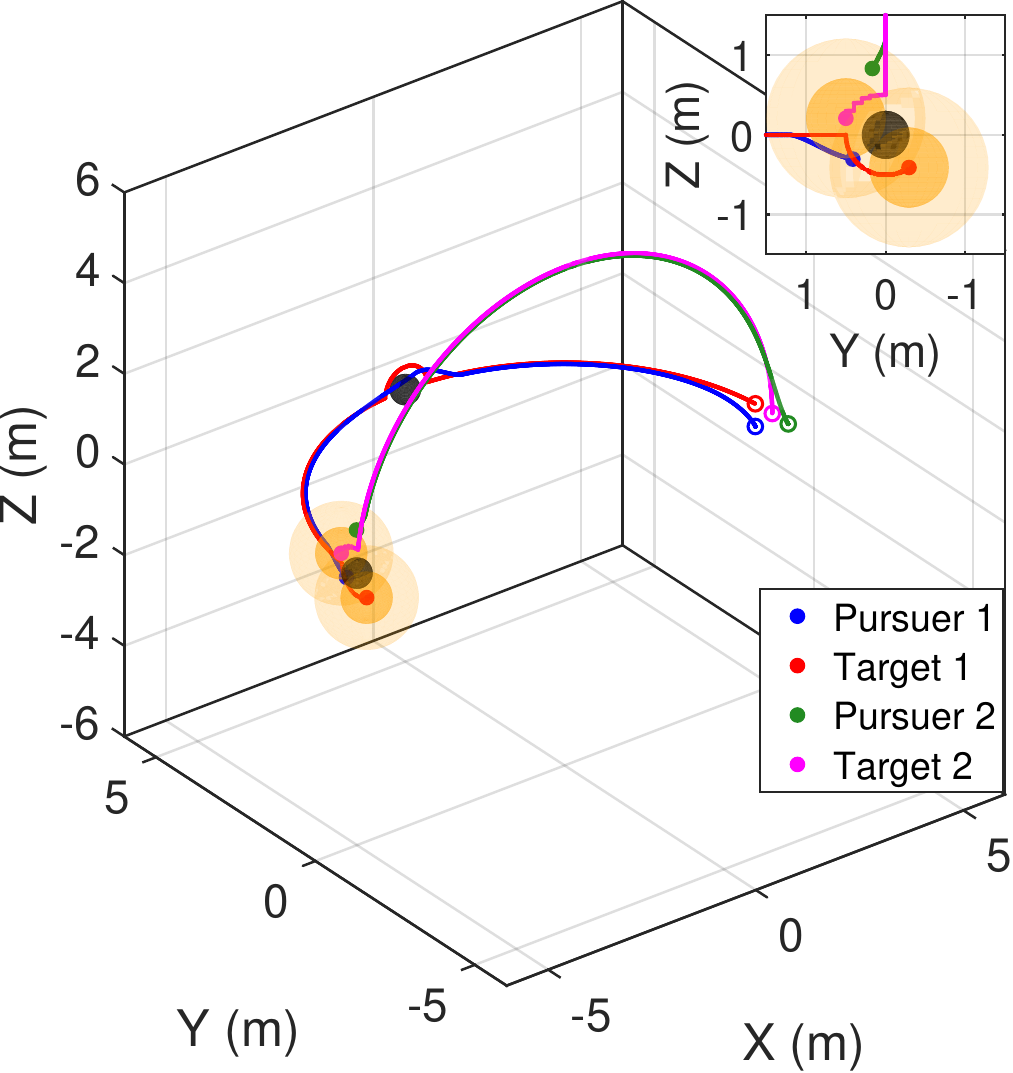} 
        \subcaption{TRUST-UP at $t=73$s.
        }\label{csrl-a-yuan} % 第四张PDF
    \end{minipage}
    \hfill
    \begin{minipage}{0.3\textwidth}
        \centering
        \includegraphics[width=\linewidth]{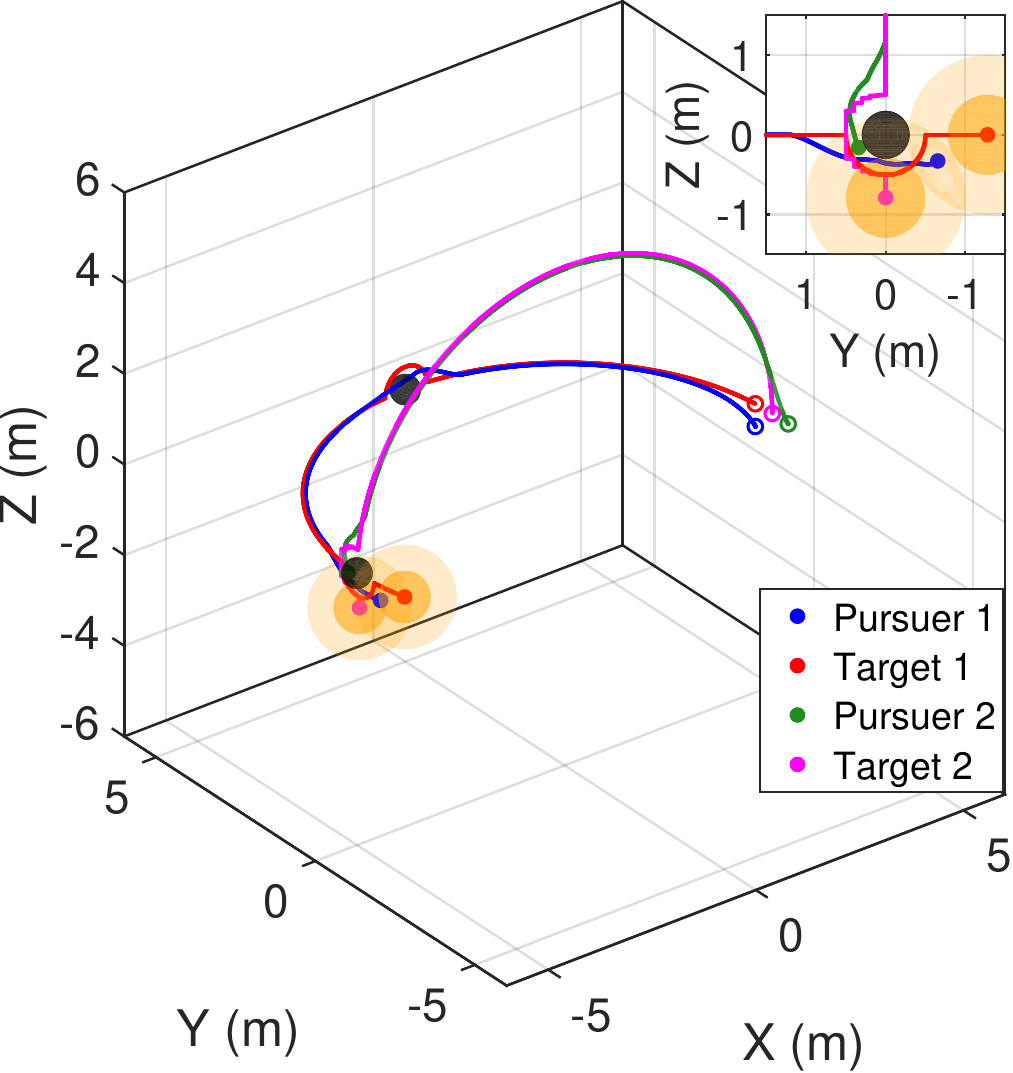} 
        \subcaption{TRUST-UP at $t=77$s.}\label{csrl-b-yuan}
    \end{minipage}
    \hfill
    \begin{minipage}{0.3\textwidth}
        \centering
        \includegraphics[width=\linewidth]{2Fig2_revised.pdf} 
        \subcaption{TRUST-UP at $t=600$s.}\label{csrl-c-yuan}
    \end{minipage}
    \caption{Snapshots of the TRUST-UP algorithm at  $t=73$s, $77$s, and $600$s during the first experiment where the targets perform circular maneuvers. The red and magenta trajectories represent two pursuer UAVs, while the blue and green trajectories represent two target UAVs. The dark yellow spheres indicate the target safety radius $r_i=0.5$, and the light yellow spheres denote the sensing range between the pursuer and its target $R_i=1.0$. Black spheres represent static obstacles with a radius of $0.3$. Hollow circles in corresponding colors represent the initial positions of each UAV. }
    \label{fig:zhengti1-yuan}
\end{figure*}
\begin{figure*}[!htb]
    \centering
     \hspace*{0.4cm}
    % 第二行的2张图片
    \begin{minipage}{0.3\textwidth}
        \centering
        \includegraphics[width=\linewidth]{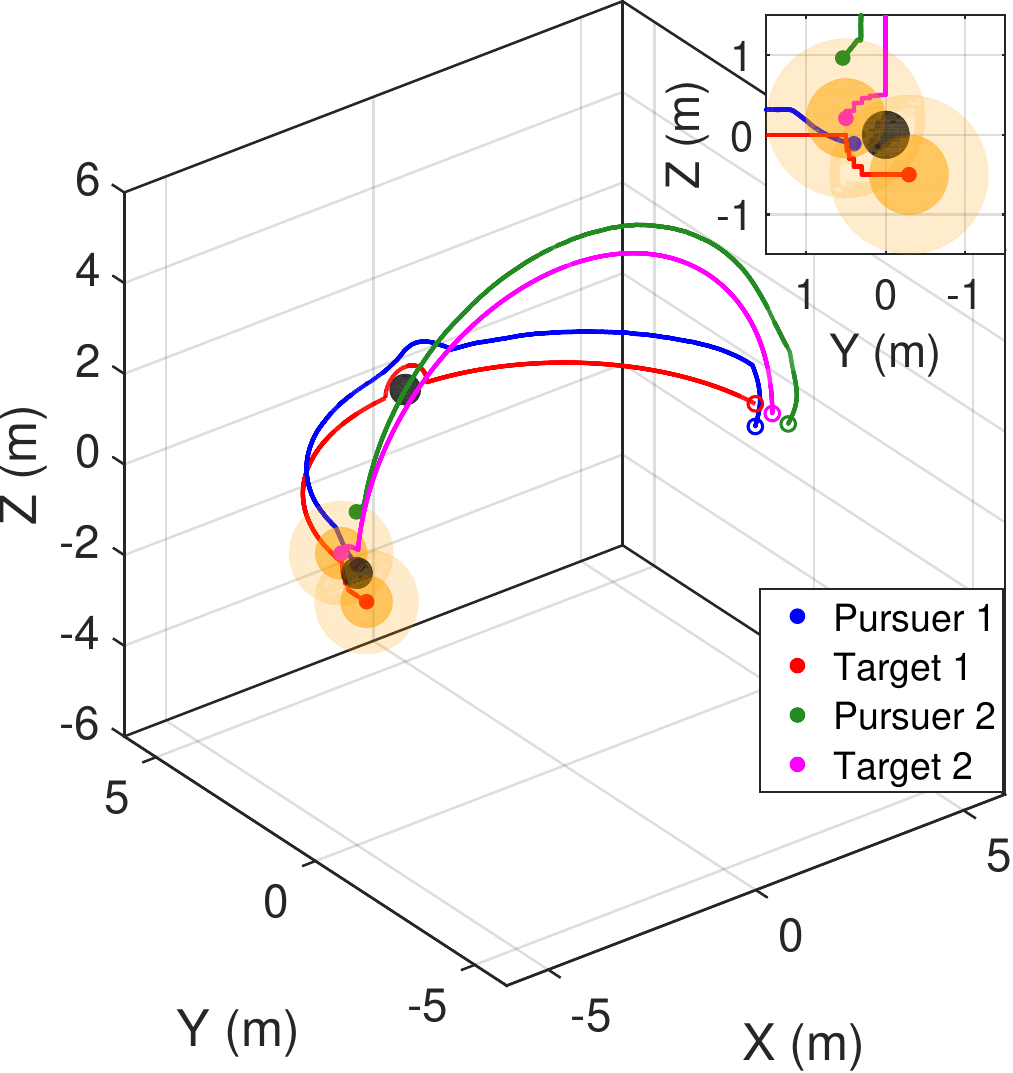} 
        \subcaption{SAC at $t=73$s}\label{onlyrl-a-yuan} % 第四张PDF
    \end{minipage}
    \hfill
    \begin{minipage}{0.3\textwidth}
        \centering
        \includegraphics[width=\linewidth]{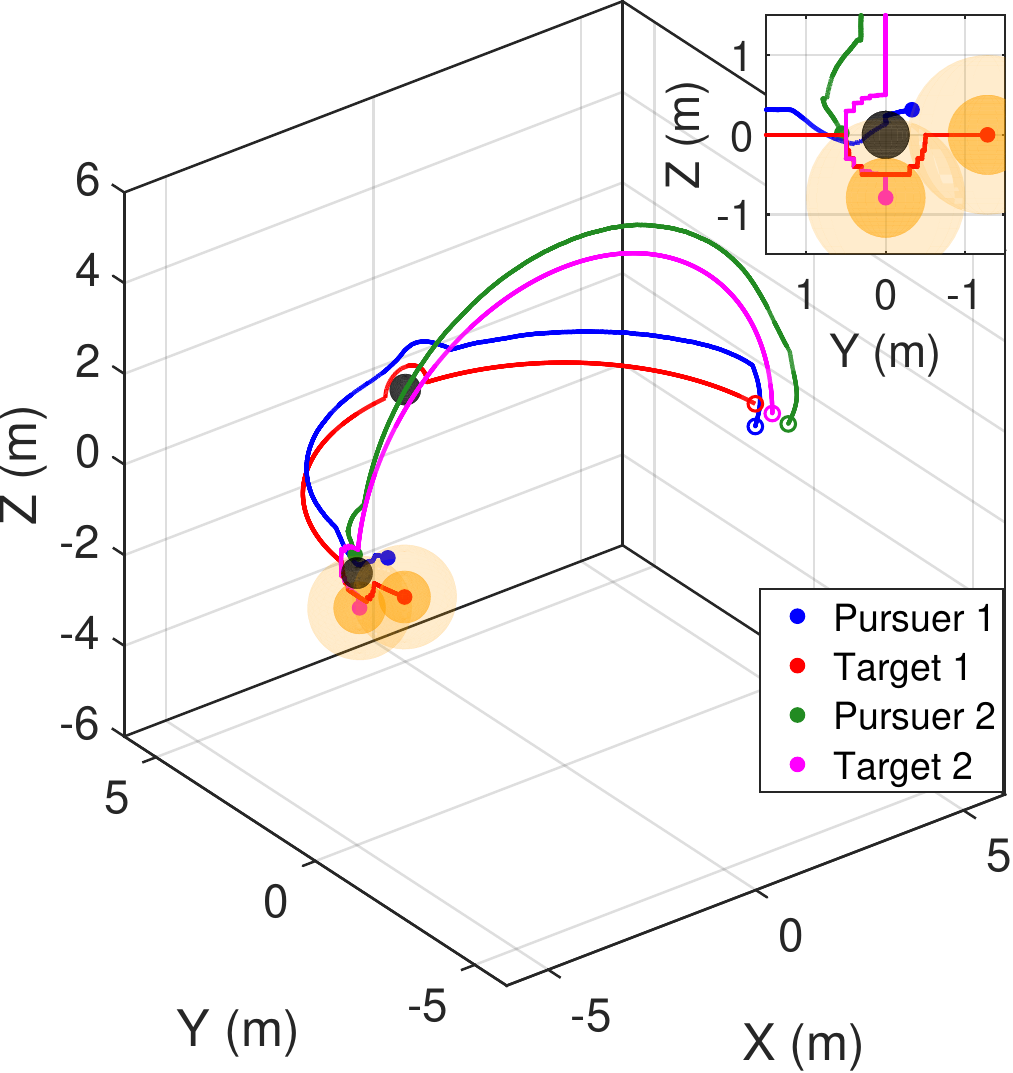} 
        \subcaption{SAC at $t=77$s}\label{onlyrl-b-yuan}
    \end{minipage}
    \hfill
    \begin{minipage}{0.3\textwidth}
        \centering
        \includegraphics[width=\linewidth]{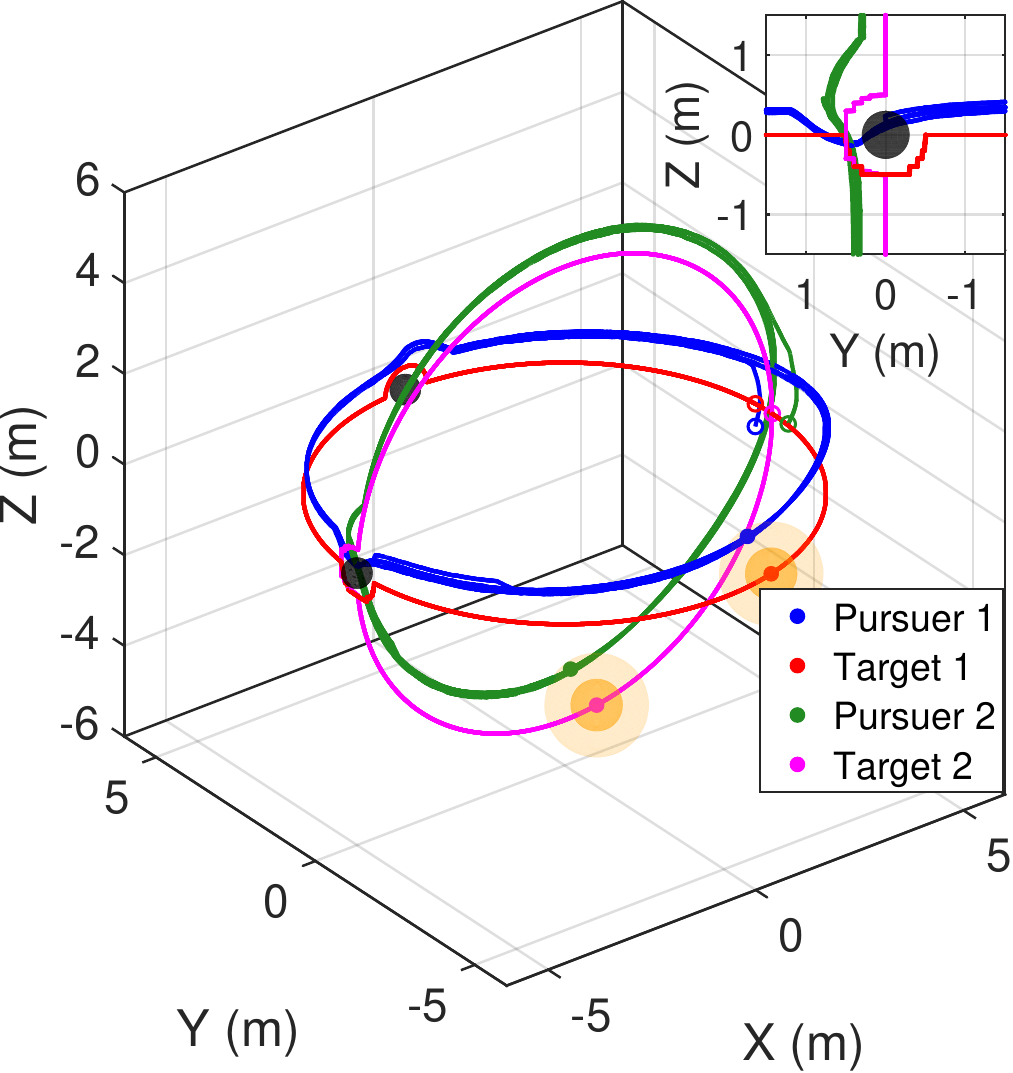} 
        \subcaption{SAC at $t=600$s}\label{onlyrl-c-yuan}
    \end{minipage}
    \caption{Snapshots of the only use SAC for pursuit UAVs at $t=73$s, $77$s, and $600$s where the targets perform circular maneuvers.}
    \label{fig:zhengti-onlyrl1-yuan}
\end{figure*}
\begin{figure*}[!htb]
    \centering
    % 第一行的2张图片
    \begin{minipage}{0.45\textwidth}
        \centering
        \includegraphics[width=\linewidth]{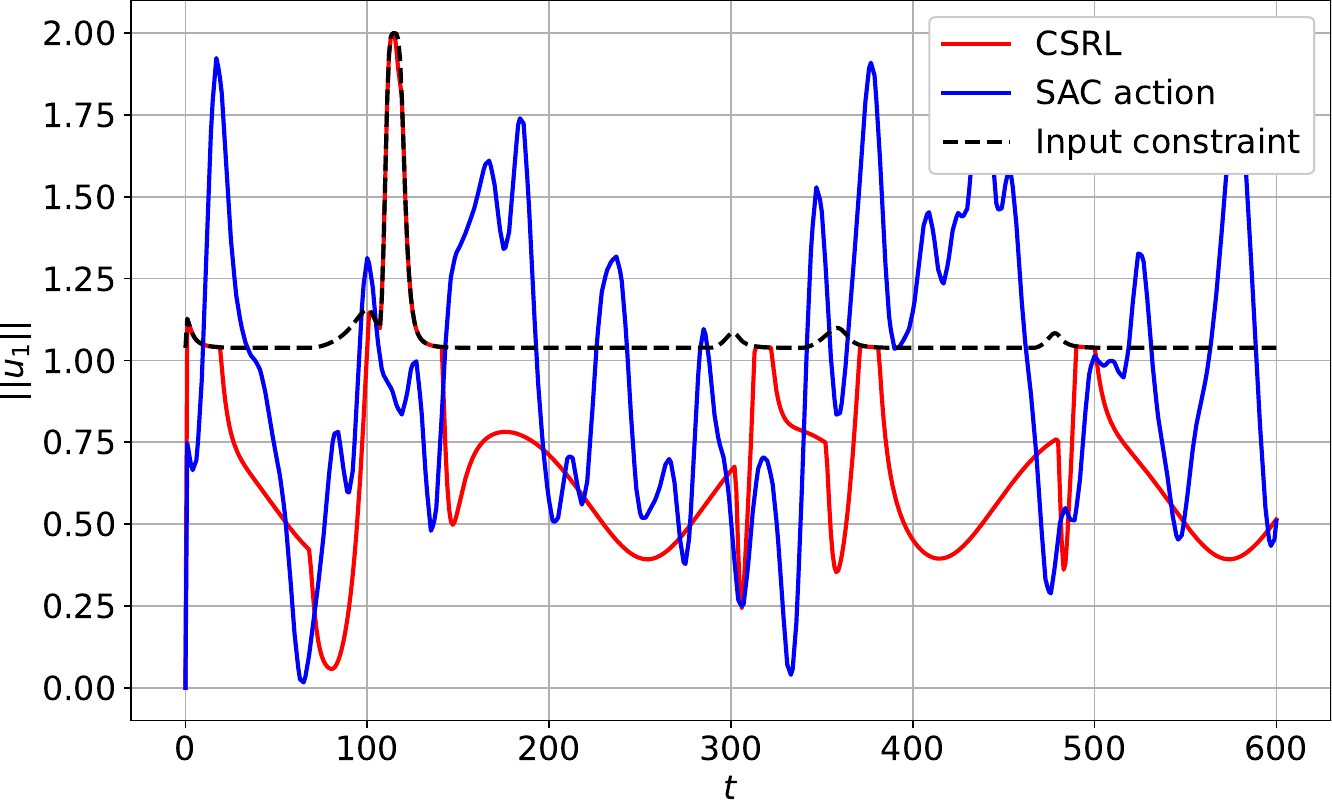} 
        \subcaption{The control input $u_1$ of TRUST-UP and RL only with time-varying input constraint $\kappa$.}\label{plot-a-yuan}
    \end{minipage}
    \hspace{0.05\textwidth} % 增加水平间距
    \begin{minipage}{0.45\textwidth}
        \centering
        \includegraphics[width=\linewidth]{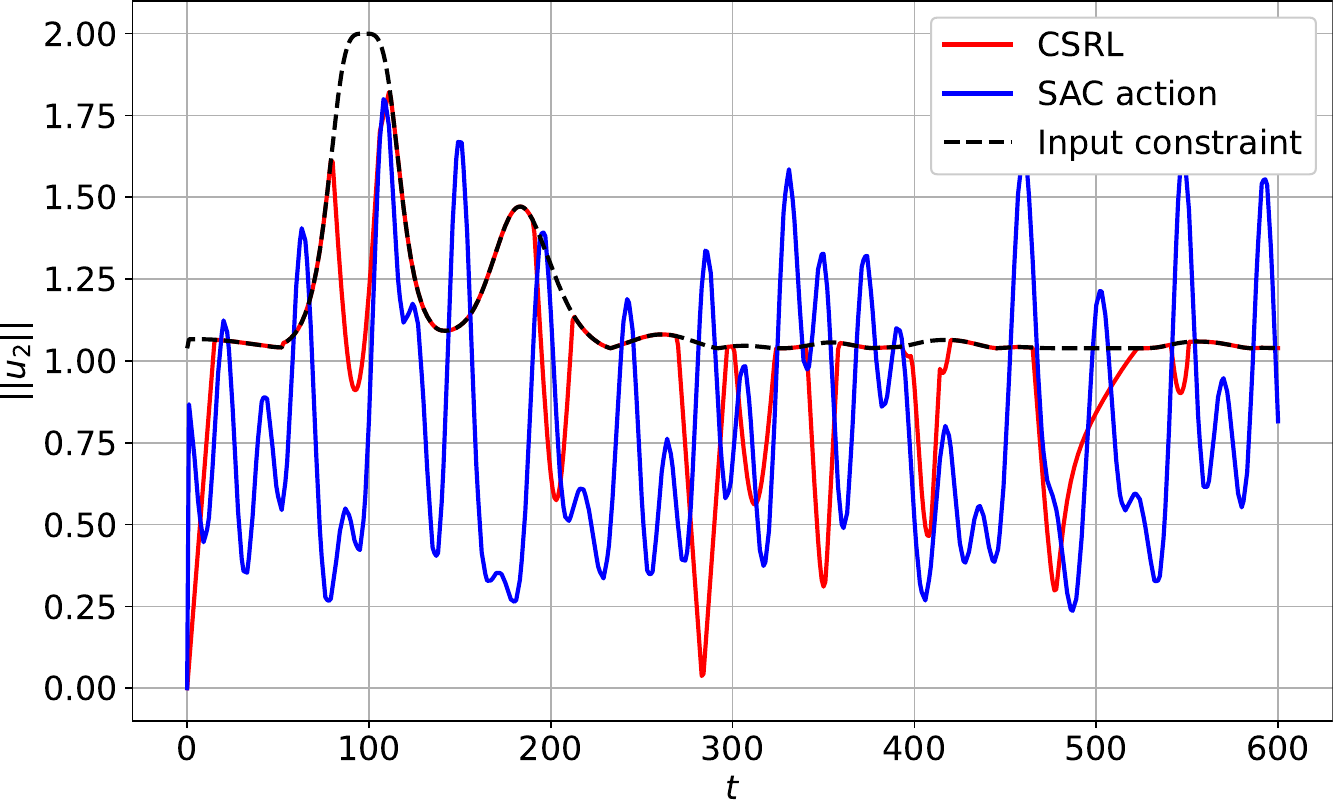} 
        \subcaption{The control input $u_2$ of TRUST-UP and SAC only with time-varying input constraint $\kappa$.}\label{plot-b-yuan}
    \end{minipage}
    
    % 第二行的2张图片
    \vspace{0.4cm} % 垂直间距
    \begin{minipage}{0.45\textwidth}
        \centering
        \includegraphics[width=\linewidth]{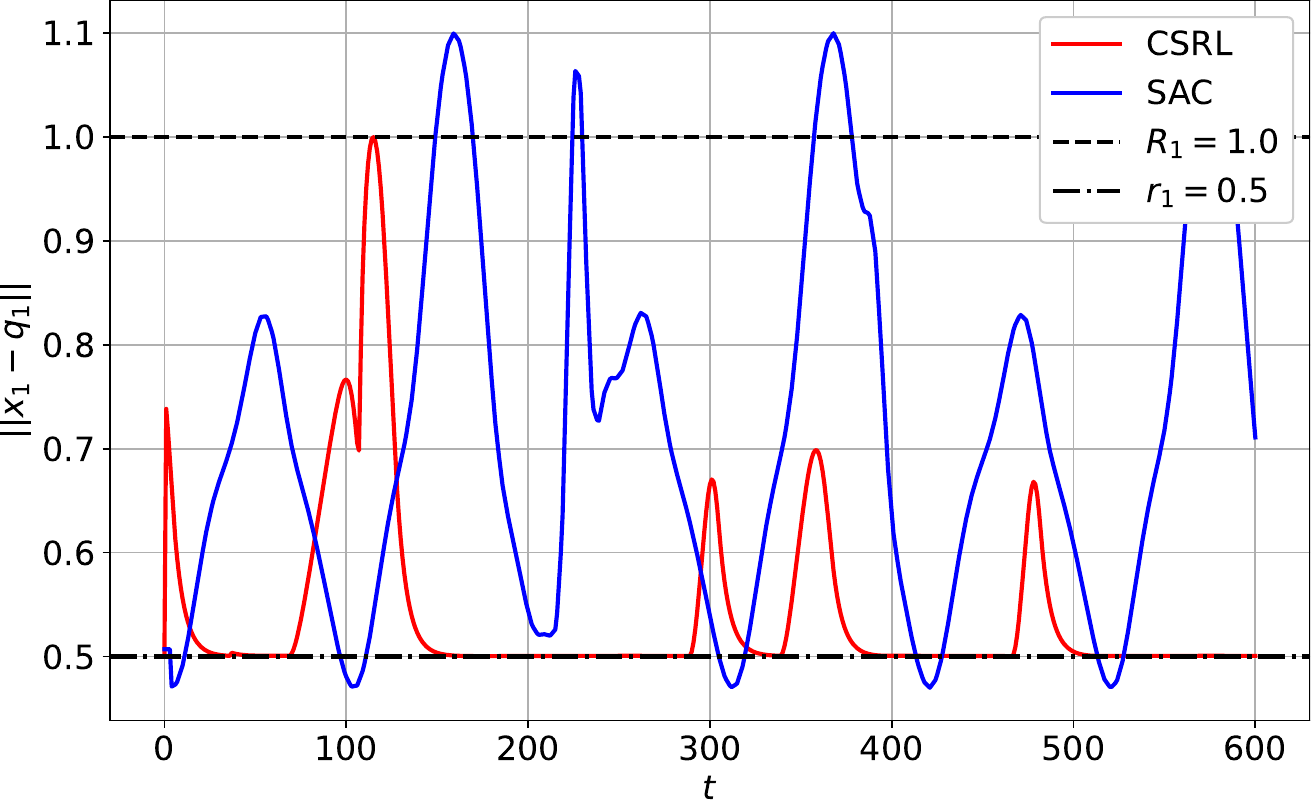} 
        \subcaption{The distance $\norm{x_1-q_1}$ of TRUST-UP and SAC only with safe collision radius $r_1 = 0.5$ and safe sensing radius $R_1=1.0$.}\label{plot-c-yuan}
    \end{minipage}
    \hspace{0.05\textwidth} % 增加水平间距
    \begin{minipage}{0.45\textwidth}
        \centering
        \includegraphics[width=\linewidth]{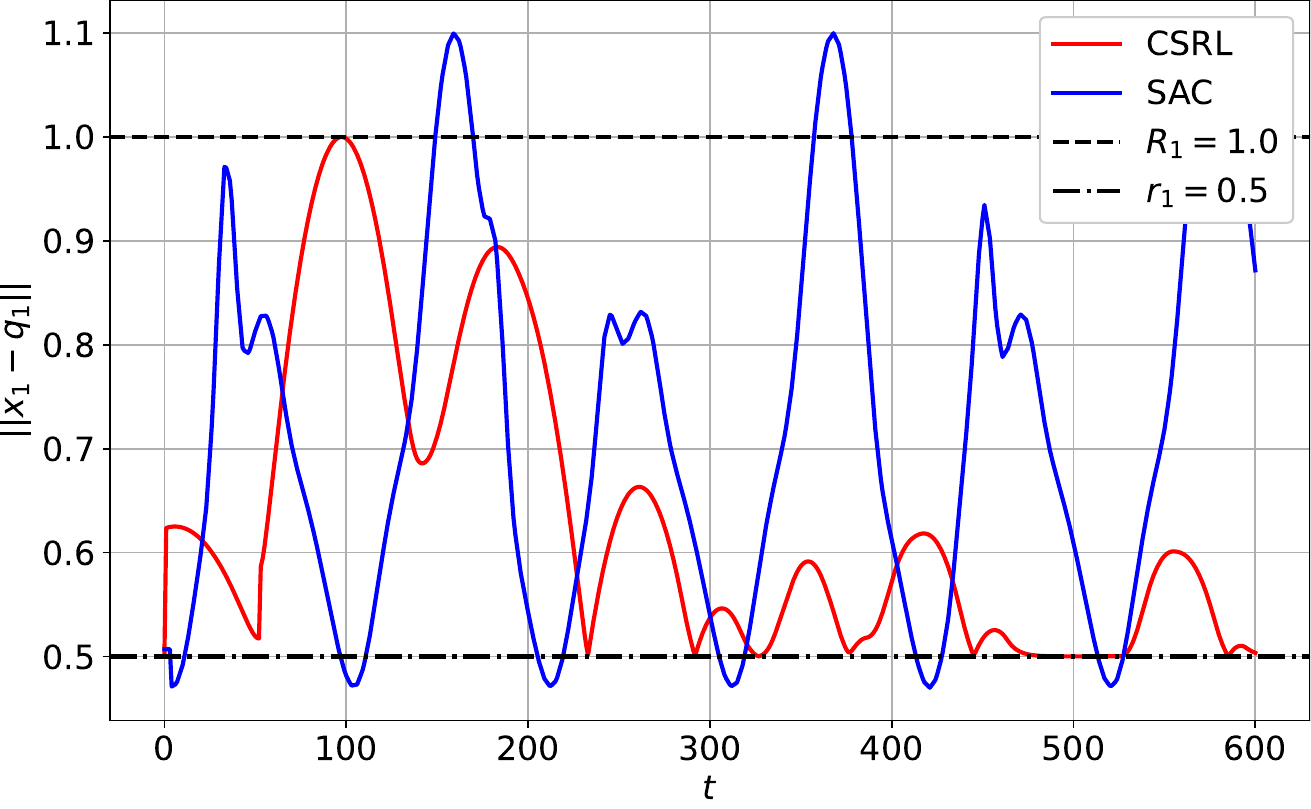} 
        \subcaption{The distance $\norm{x_2-q_2}$ of TRUST-UP and SAC only with safe collision radius $r_2 = 0.5$ and safe sensing radius $R_2=1.0$.}\label{plot-d-yuan}
    \end{minipage}
    \caption{Comparison of the TRUST-UP algorithm and SAC-only control in terms of control input $u_i$ and relative distance $\norm{\zeta_i}=\norm{x_i-q_i}$ for pursuer UAVs $1$ and $2$.}
    \label{fig-prcbf-yuan}
\end{figure*}

\begin{figure*}[!htb]
    \centering
     \hspace*{0.4cm}
    % 第二行的2张图片
    \begin{minipage}{0.3\textwidth}
        \centering
        \includegraphics[width=\linewidth]{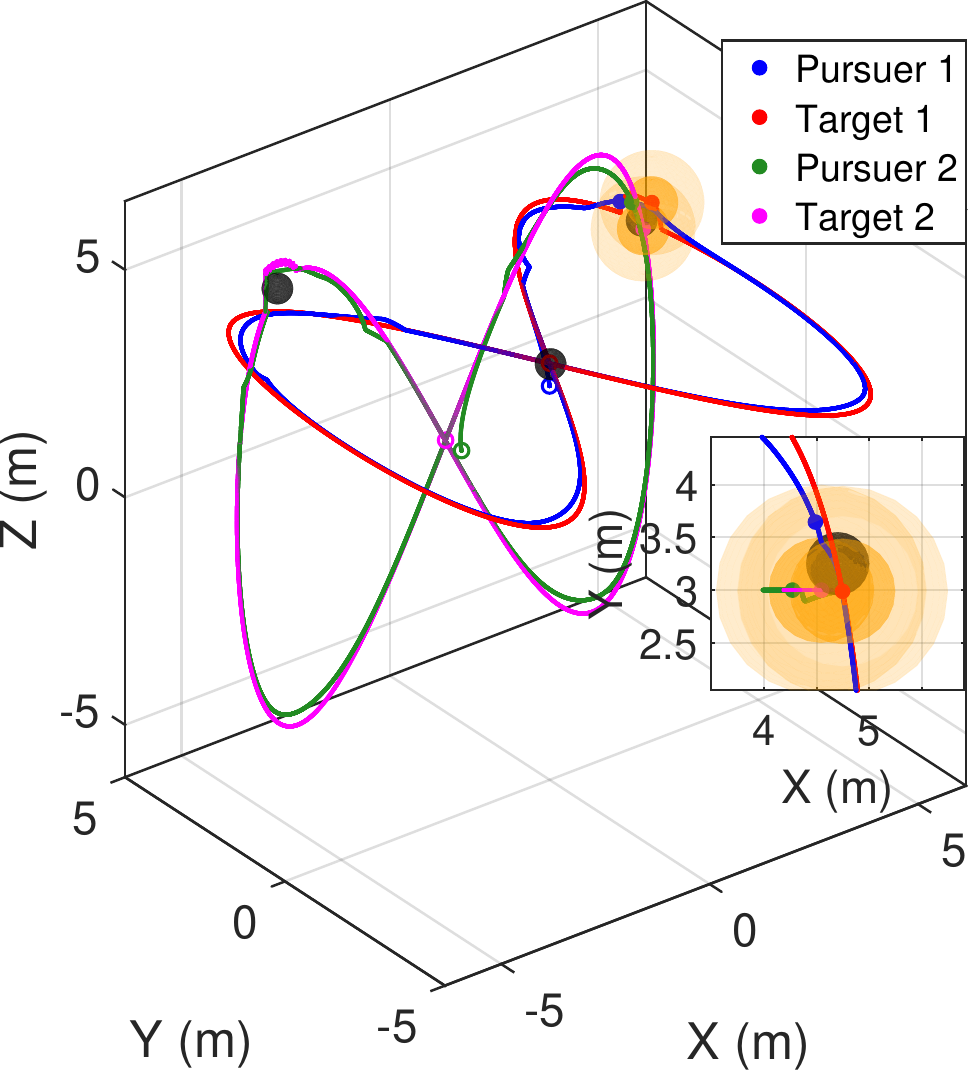} 
        \subcaption{TRUST-UP at $t=152$s.
        }\label{csrl-a} % 第四张PDF
    \end{minipage}
    \hfill
    \begin{minipage}{0.3\textwidth}
        \centering
        \includegraphics[width=\linewidth]{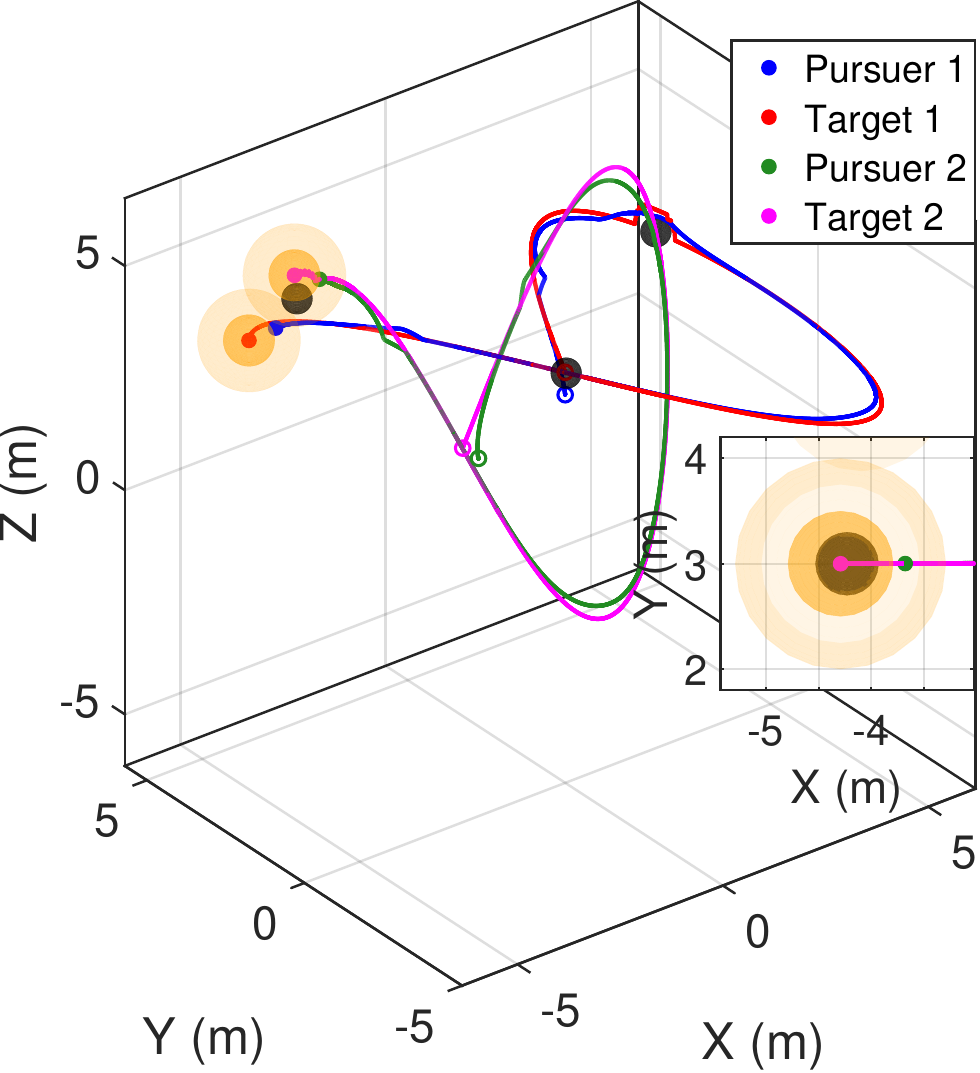} 
        \subcaption{TRUST-UP at $t=401$s.}\label{csrl-b}
    \end{minipage}
    \hfill
    \begin{minipage}{0.3\textwidth}
        \centering
        \includegraphics[width=\linewidth]{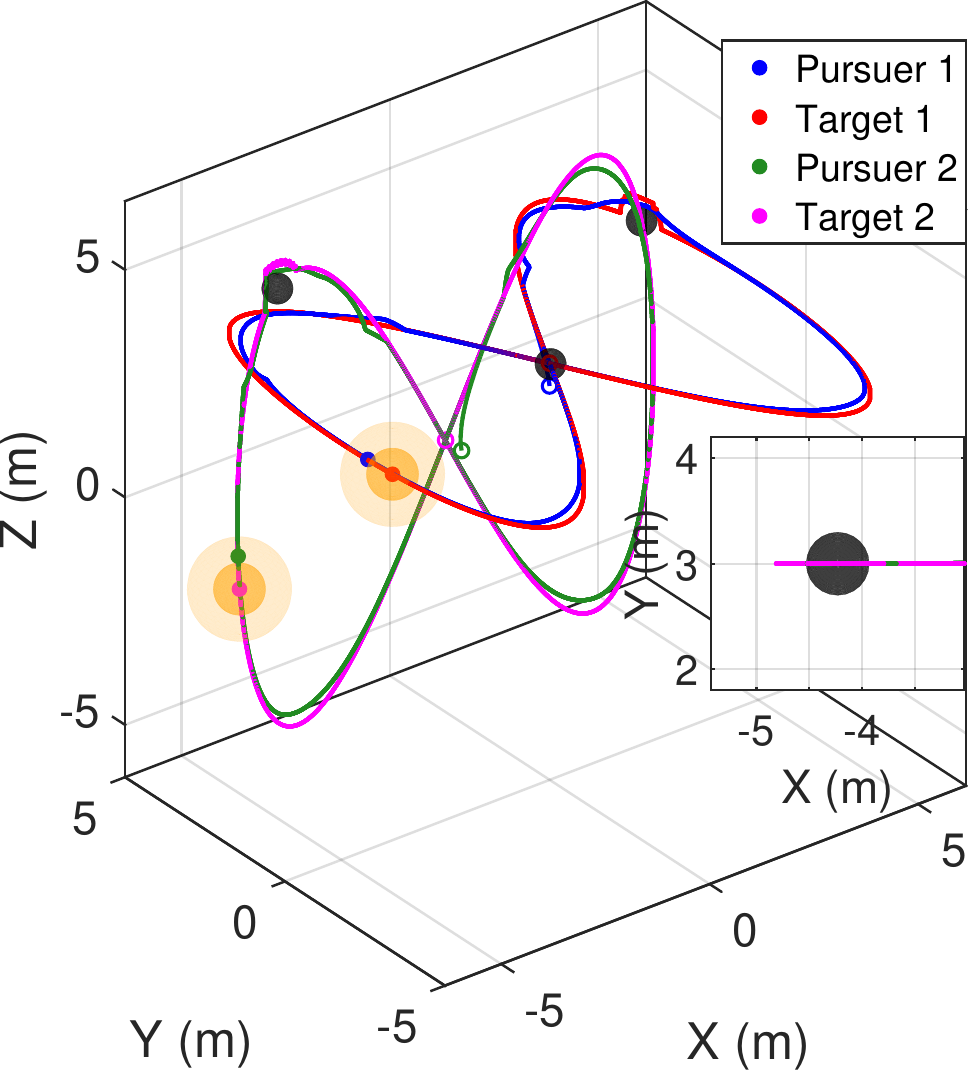} 
        \subcaption{TRUST-UP at $t=600$s.}\label{csrl-c}
    \end{minipage}
    \caption{Snapshots of the TRUST-UP algorithm at $t=152$s, $401$s, and $600$s where the targets perform ``figure-8" maneuvers.}
    \label{fig:zhengti1}
\end{figure*}
\begin{figure*}[!htb]
    \centering
     \hspace*{0.4cm}
    % 第二行的2张图片
    \begin{minipage}{0.3\textwidth}
        \centering
        \includegraphics[width=\linewidth]{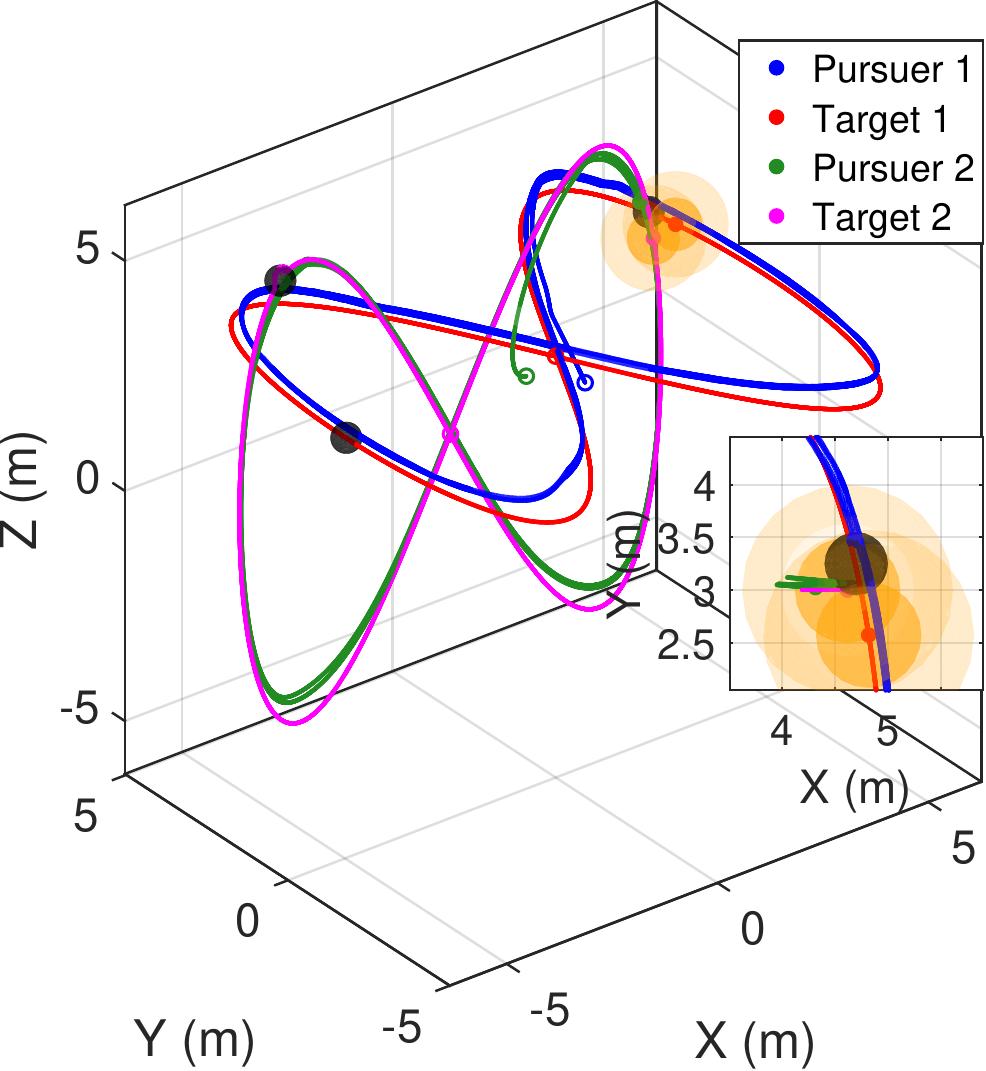} 
        \subcaption{SAC at $t=152$s}\label{onlyrl-a} % 第四张PDF
    \end{minipage}
    \hfill
    \begin{minipage}{0.3\textwidth}
        \centering
        \includegraphics[width=\linewidth]{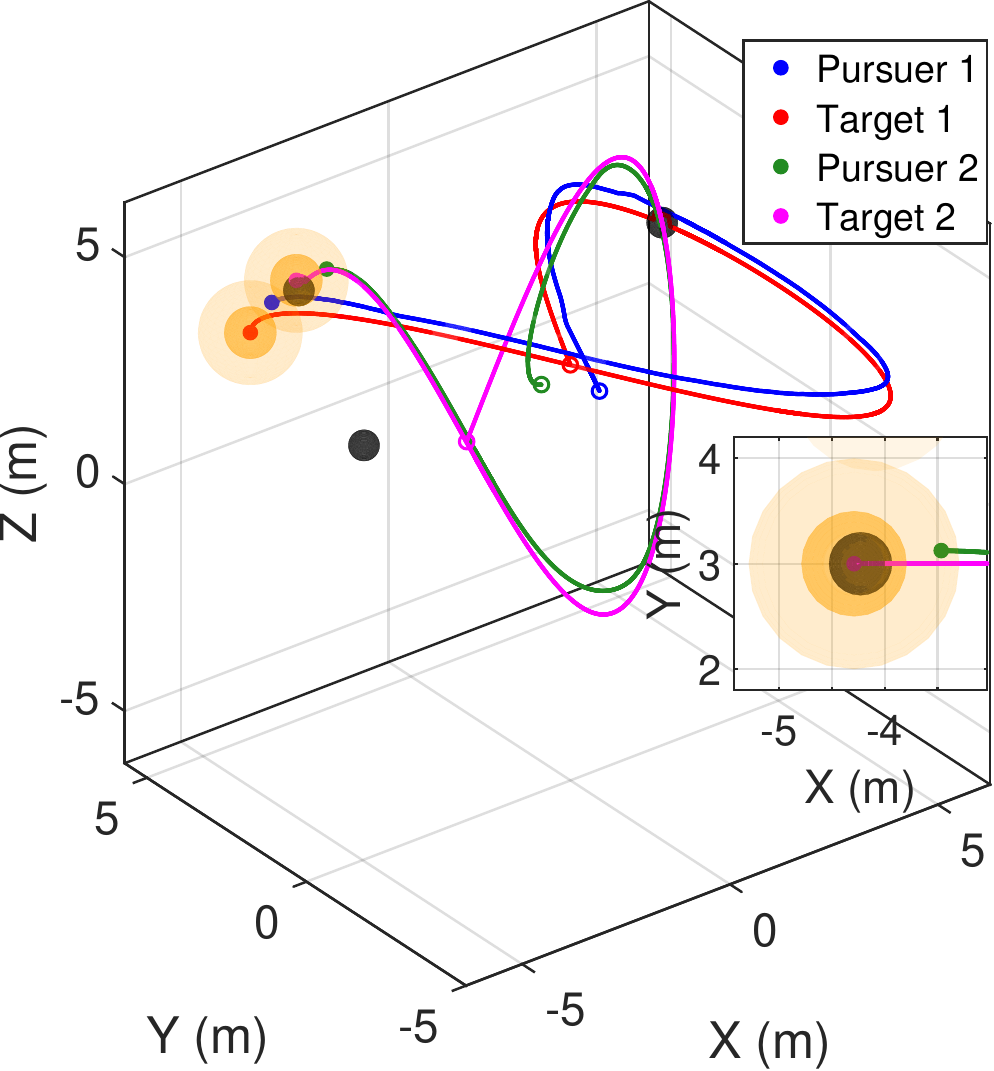} 
        \subcaption{SAC at $t=401$s}\label{onlyrl-b}
    \end{minipage}
    \hfill
    \begin{minipage}{0.3\textwidth}
        \centering
        \includegraphics[width=\linewidth]{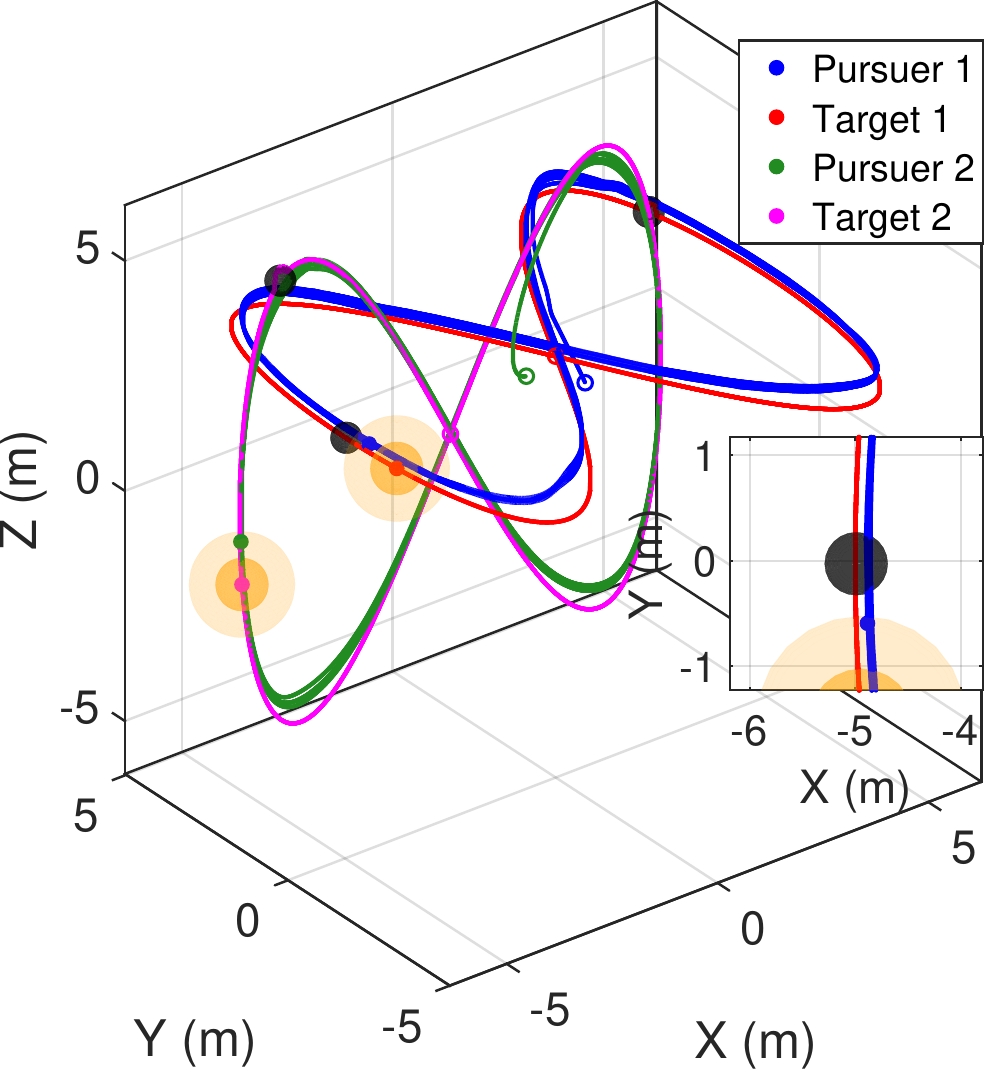} 
        \subcaption{SAC at $t=600$s}\label{onlyrl-c}
    \end{minipage}
    \caption{Snapshots of the only use SAC as the RL controller for pursuit UAVs at $t=152$s, $401$s, and $600$s where the targets perform ``figure-8" maneuvers.}
    \label{fig:zhengti-onlyrl1}
\end{figure*}
\begin{figure*}[!htb]
    \centering
    % 第一行的2张图片
    \begin{minipage}{0.45\textwidth}
        \centering
        \includegraphics[width=\linewidth]{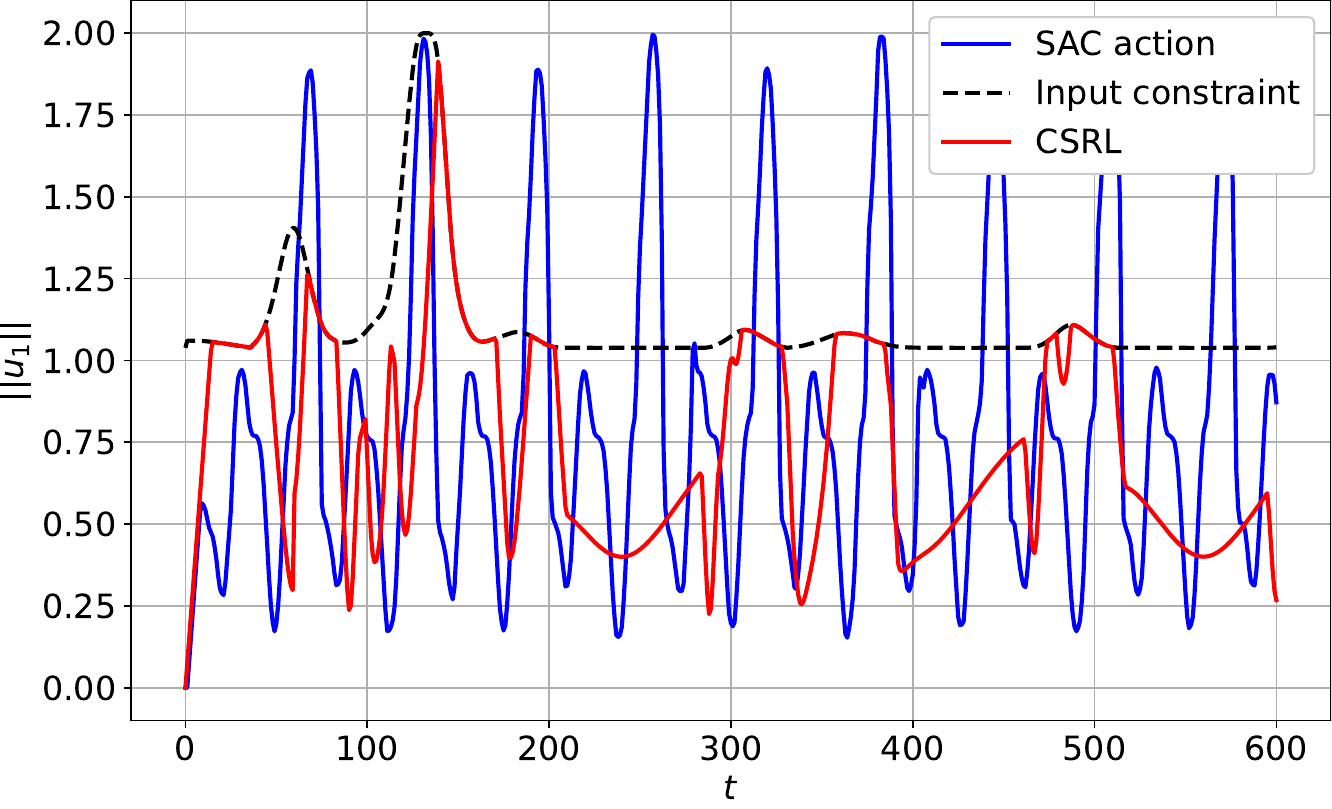} 
        \subcaption{The control input $u_1$ of TRUST-UP and SAC only with time-varying input constraint $\kappa$.}\label{plot-a}
    \end{minipage}
    \hspace{0.05\textwidth} % 增加水平间距
    \begin{minipage}{0.45\textwidth}
        \centering
        \includegraphics[width=\linewidth]{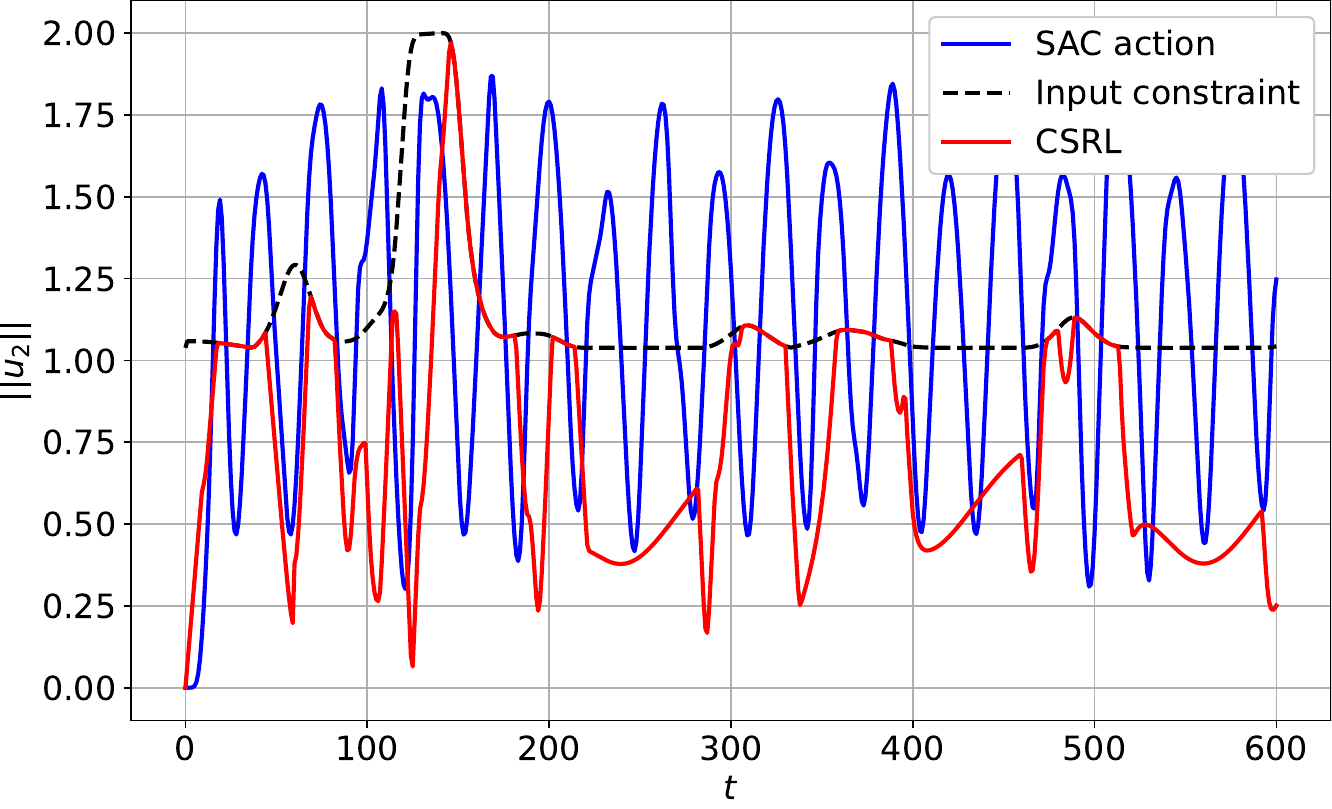} 
        \subcaption{The control input $u_2$ of TRUST-UP and SAC only with time-varying input constraint $\kappa$.}\label{plot-b}
    \end{minipage}
    
    % 第二行的2张图片
    \vspace{0.4cm} % 垂直间距
    \begin{minipage}{0.45\textwidth}
        \centering
        \includegraphics[width=\linewidth]{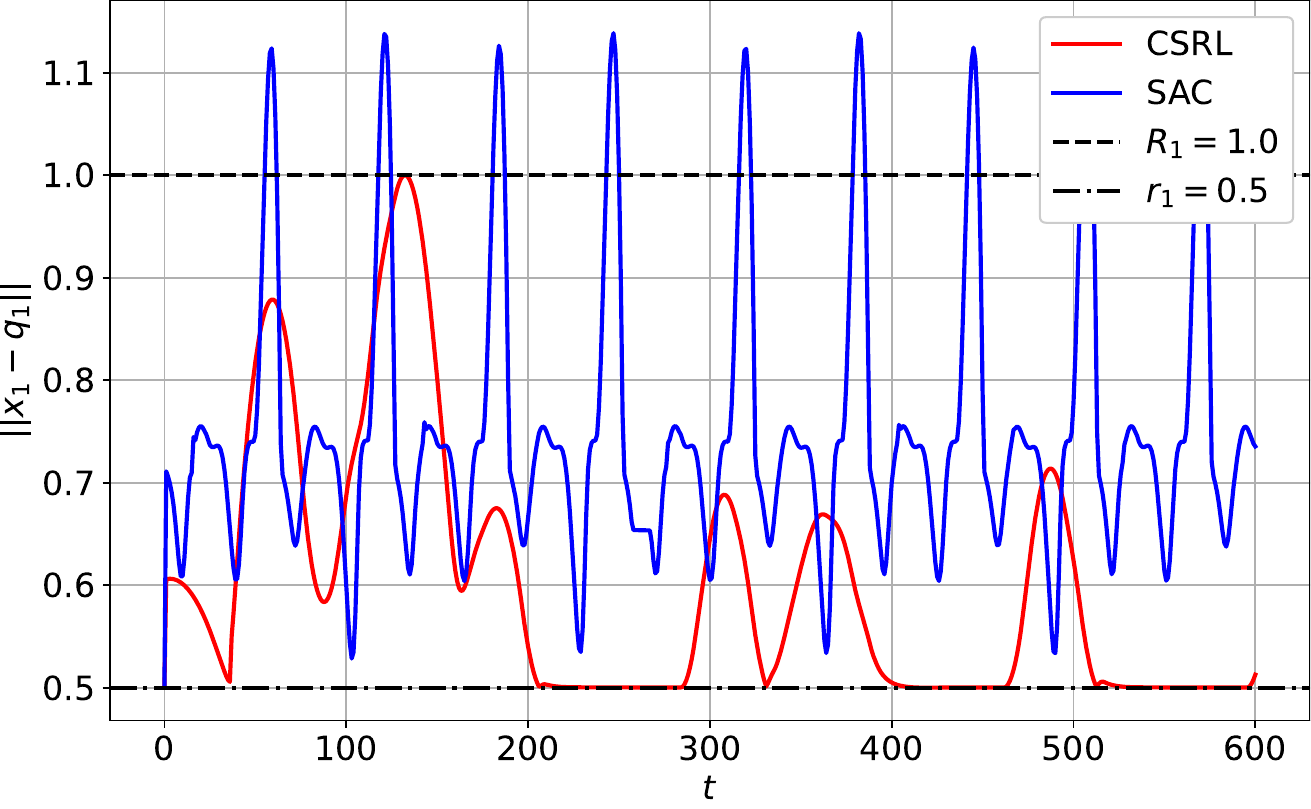} 
        \subcaption{The distance $\norm{x_1-q_1}$ of TRUST-UP and SAC only with safe collision radius $r_1 = 0.5$ and safe sensing radius $R_1=1.0$.}\label{plot-c}
    \end{minipage}
    \hspace{0.05\textwidth} % 增加水平间距
    \begin{minipage}{0.45\textwidth}
        \centering
        \includegraphics[width=\linewidth]{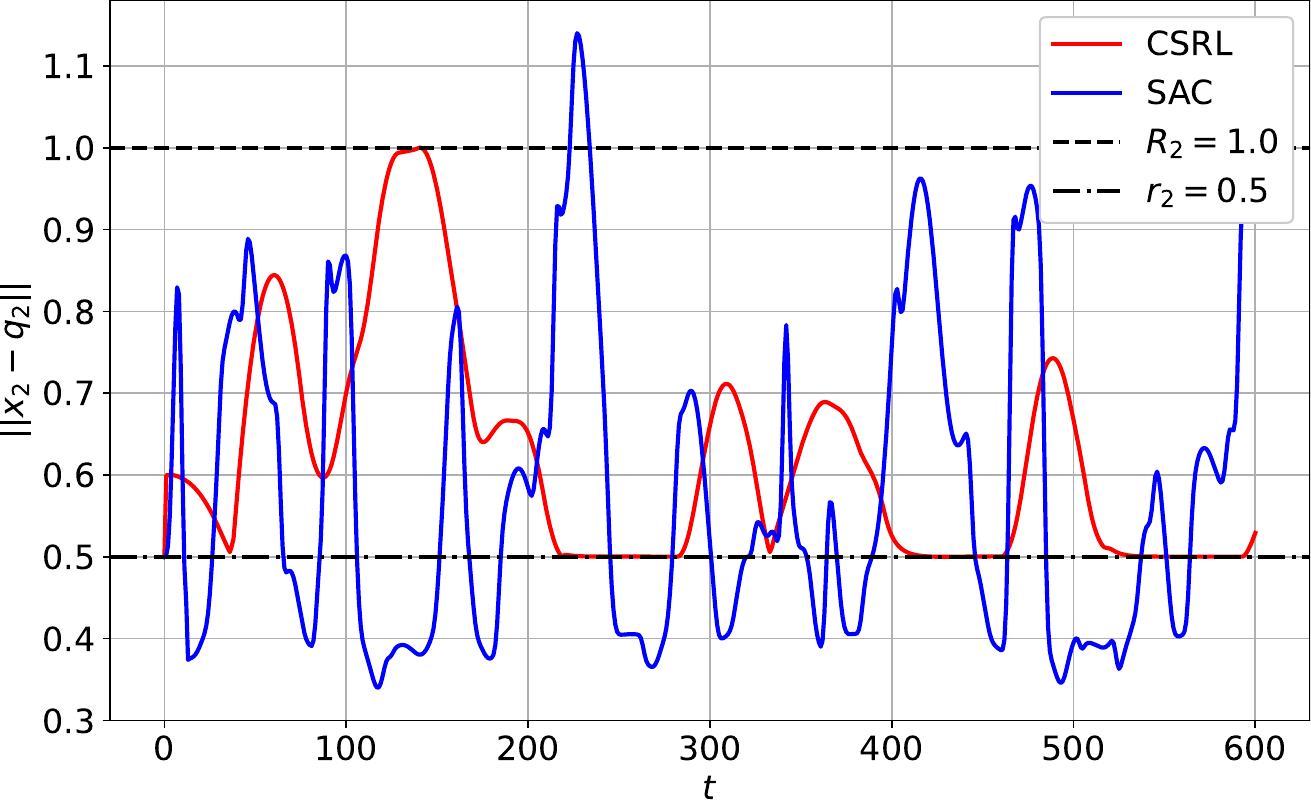} 
        \subcaption{The distance $\norm{x_2-q_2}$ of TRUST-UP and SAC only with safe collision radius $r_2 = 0.5$ and safe sensing radius $R_2=1.0$.}\label{plot-d}
    \end{minipage}
    \caption{Comparison of the TRUST-UP algorithm and SAC-only control in terms of control input $u_i$ and relative distance $\norm{\zeta_i}=\norm{x_i-q_i}$ for pursuer UAVs $1$ and $2$.}
    \label{fig-prcbf}
\end{figure*}
In this section, we design a set of numerical simulations involving two target-pursuit multicopter experiments to verify the effectiveness and safety of the proposed method. Specifically, two follower UAVs operate in the same environment, each tasked with tracking a corresponding target UAV, while navigating around obstacles and subject to unknown disturbances. To achieve this, we adopt the Soft Actor-Critic (SAC)~\cite{haarnoja1861soft}, an off-policy, model-free reinforcement learning algorithm, as the nominal control input for the followers.

In the model-free RL framework, the control action $\pi_i$ is generated by the SAC policy for the $i$-th pursuer, which is trained to enable a follower UAV to maintain a specified tracking range from its target in an obstacle-free environment without disturbances. The objective is to develop an RL control strategy for tracking that can later be integrated with safety constraints. To train the SAC policy for the pursuer’s control actions, the reward function $r(x_i, \pi_i)$ is designed to encourage the $i$-th pursuer's position $x_i$ remaining within a desired distance range $[r_1, R_1]$ to the target position $p_0$. The reward function is formulated as:
\begin{equation}\label{reward f}
r(x_i, \pi_i) = 
\begin{cases} 
0.1 & \text{if } r_i \leq \norm{\zeta_i} \leq R_i, \\
-0.1 |\norm{\zeta_i} - r_i| & \text{if } \norm{\zeta_i} < r_i, \\
-0.1 |\norm{\zeta_i}- R_i| & \text{if } \norm{\zeta_i} > R_i.
\end{cases}
\end{equation}
The reward function~\eqref{reward f} is designed to capture the relationship between tracking error and reward, assigning larger rewards for small errors and diminishing rewards as the error increases. This structure prevents the reward from approaching zero for large errors, aiding numerical stability and discouraging significant deviations from the target. {Note that the policies in~\eqref{reward f} may suffer from suboptimal convergence due to non-smooth reward gradients. However, the proposed safety filter decouples system safety from RL convergence. Even if the SAC agent outputs suboptimal commands, the deterministic safety filter actively intervenes to guarantee the safety of the human trust radii.}
The SAC policy is trained over $1500$ episodes and the time interval between each step is 0.1$s$ for each dimension of pursuer's acceleration. 
Figure~\ref{fig_training} illustrates the average return training curves for the SAC algorithms. The results indicate that training a SAC-based target pursuit controller fails due to the challenges of balancing multiple safety constraints during target maneuvers, which significantly increases the training complexity and can lead to failure. To ensure safety, the CBF-QP-based algorithm~\eqref{QP}, as outlined in Algorithm~\ref{alg:safety_rl_execution}, is employed for two pursuit UAVs after training, enabling safe tracking of their respective targets while accounting for disturbances, obstacles, sensing range and input saturation. {The continuous-time safety filter~\eqref{QP} operates discretely in the simulation. At each $\Delta t = 0.1$ control step, numerical integration maps the auxiliary input to the physical control signal.}

For simplicity, we set the uncertain viscous friction functions $Y_i \theta = \sin(x_i)$, $Z_i \xi = \cos(x_i)$, with uncertain coefficients $\theta = \xi = 1$ for all pursuit UAVs. Two sphere obstacles are placed and their
specific positions of center $p_1 = [4.70, 3.25, 3.00]$; $p_2 = [-4.20, 3.00, 4.75]$. The safe radius is set as $r_i=0.5$ and $R_i=1$ for all $i\in\mathcal{I}_x$. The total time of this simulation is $600$ steps and the time step is set as $0.1$.
The movement law of the first target $p_{0,1}$, and second target $p_{0,2}$ is defined as
\begin{equation}\label{target dynamics}
\begin{aligned}
    \ddot p_{0,1} &= \ddot p_{r1} + (p_{r1}-p_{0,1}) +  (p_{r1} - \dot p_{0,1}) + U_1(x), \\
    \ddot p_{0,2} &= \ddot p_{r2} + (p_{r2}-p_{0,2}) +  (p_{r2} - \dot p_{0,2})+ U_2(x),
\end{aligned}
\end{equation}
where $U_j(x)= \sum_{i=1}^2 \left( \frac{1}{\norm{p_0-p_i}} - 0.1 \right) \frac{p_{0,j}-p_i}{\norm{p_0-p_i}^3}$, $j\in\{1,2\}$, is a potential field term for target UAV collision-avoidance to obstacle $p_1$ and $p_2$, and
$p_{r1}$, $p_{r2}$ in~\eqref{target dynamics} are the reference signal for target UAVs. In our first experiment, the targets follow a large circular maneuver, with their positions updated as:
$p_{r1}=[5\sin(0.1t), 5\cos(0.1t), 0]$ and $p_{r2}=[5\sin(0.1t), 0, 5\cos(0.1t)]$
In our second experiment, the targets follow a ``figure-8" reference signal~\cite{o2022neural}, given by: given by $p_{r1}=[5\sin(0.1t), 5\sin(0.2t), 3]$ and $p_{r2}=[5\sin(0.1t), 3, 5\sin(0.2t)]$. 

In the first experiment, we evaluated the proposed TRUST-UP algorithm~\ref{alg:safety_rl_execution} with two targets following large circular maneuvers as shown in Figure~\ref{fig:zhengti1-yuan}. The targets incorporated a potential field term in~\eqref{target dynamics} for collision avoidance with obstacles. From Figure~\ref{csrl-a}, both target-pursuer pairs converge near $[-5;0;0]$ , where a radius $0.3$ obstacle is present. The targets perform large-angle maneuvers guided by the potential field term in~\eqref{target dynamics}. Close-up views indicate that all pursuers maintain safe distances from obstacles and other UAVs, while retaining sensing of their corresponding targets at $t=73$s. Then from Figure~\ref{csrl-b-yuan}, the pursuers successfully navigate around the obstacle while maintaining sensing safety with their respective targets. In Figure~\ref{csrl-c-yuan} the trajectories of both pursuers remain relatively stable under external disturbances, ensuring reliable tracking of their targets. Figure~\ref{fig:zhengti-onlyrl1-yuan} presents a comparison with the SAC-only method. In Figure~\ref{onlyrl-a-yuan}, the blue pursuer collides with the obstacle, highlighting the inability of SAC to guarantee collision safety. In Figure~\ref{onlyrl-b-yuan}, both pursuers fail to maintain sensing safety, as the distances between pursuers and their respective targets exceed the sensing range $R_i=1$ (indicated by the light yellow region). From Figure~\ref{onlyrl-c-yuan} we observe that under external disturbances, the SAC algorithm demonstrates unstable tracking behavior, failing to reliably pursue the targets. These results demonstrate the superior performance of the TRUST-UP algorithm in ensuring collision safety, sensing safety, and tracking stability, even in challenging scenarios with obstacles and disturbances, compared to the SAC-only method. 
% Figure~\ref{fig-prcbf-yuan} is the control inputs and states plot corresponding to the scenario in Figure~\ref{fig:zhengti1-yuan} and Figure~\ref{fig:zhengti-onlyrl1-yuan}, where Figure~\ref{plot-a-yuan} and Figure~\ref{plot-b-yuan} demonstrates that, compared to the SAC algorithm, the proposed TRUST-UP algorithm ensures that the control inputs of both pursuers consistently satisfy the input constraints in~\eqref{C3} throughout the target pursuit process. {Here, the intervals where (|u_i|) is regulated by the dashed input-constraint boundary correspond to the activation of the CBF-QP safety filter, during which no evident high-frequency oscillation is observed.} 
Figure~\ref{fig-prcbf-yuan} presents the control inputs and states corresponding to the scenario in Figure~\ref{fig:zhengti1-yuan} and Figure~\ref{fig:zhengti-onlyrl1-yuan}. Figure~\ref{plot-a-yuan} and Figure~\ref{plot-b-yuan} shows that, compared with SAC, the proposed TRUST-UP algorithm keeps the control inputs of both pursuers within the prescribed bounds in~\eqref{C3} throughout the pursuit process, {and no evident high-frequency oscillation is observed when the inputs are limited by the dashed constraint boundary.}
Figure~\ref{plot-c-yuan} and Figure~\ref{plot-d-yuan} illustrate that, under the TRUST-UP algorithm, both pursuers maintain sensing safety~\eqref{C2} and collision safety~\eqref{C1} with their respective targets, which are not guaranteed under the SAC algorithm.

To evaluate the algorithm's performance in more dynamic environments, we conducted a second set of experiments where the target UAVs followed ``figure-8" trajectory. These trajectories introduce varying curvature and dynamic changes, providing a more challenging test for the TRUST-UP algorithm. Figure~\ref{fig:zhengti1} shows our proposed algorithm~\ref{alg:safety_rl_execution} by two groups of pursuit UAVs tracking the corresponding target UAVs with ``figure-8-shape" reference signal. It can be obtained that our algorithm ensures the safe sensing range $[r_i, R_i]$ and collision-avoidance of both follower $1$ to target $1$ and follower $2$ to target $2$. 
From Figure~\ref{csrl-a}, at $t=152$s,  the two pursuer UAVs maintain a safe distance from each other while tracking their respective targets $\norm{x_1-x_2}=0.5162>0.5$.  In contrast, Figure~\ref{onlyrl-a} shows that the SAC-only method fails to maintain the safe collision radius between the pursuers ($0.2822<0.5)$. In Figure~\ref{csrl-b}, at $t=401$s, the two pursuer UAVs not only avoid obstacles but also maintain sensing of their respective targets. However, in Figure~\ref{onlyrl-b}, it is evident that both pursuer $1$ and pursuer $2$ lose sensing of their targets, highlighting the limitation of the RL-only method. Finally, comparing Figure~\ref{csrl-c} and Figure~\ref{onlyrl-c}, it is clear that under external disturbances $\theta$ and $\xi$, the proposed TRUST-UP algorithm demonstrates superior tracking stability compared to the SAC-only method, ensuring reliable pursuit control even in challenging scenarios.

Figures~\ref{plot-a} and~\ref{plot-b} demonstrate that under the TRUST-UP algorithm, both pursuer $1$ and pursuer $2$ maintain their control inputs within the input constraint $\kappa$ in the second experiment. Around $t=125$s, when the targets exhibit evasive maneuvers, the input constraint temporarily increases to accommodate the stronger control effort required to ensure sensing safety. During normal tracking phases, the input constraint effectively regulates the control input, preventing unstable nominal control signals and ensuring smooth operation. Figures~\ref{plot-c} and~\ref{plot-d} show that the TRUST-UP algorithm enables both pursuer $1$ and pursuer $2$ to maintain the required safe collision radius and sensing range with their respective target UAVs. This ensures that the pursuers achieve safe and effective target tracking, even under challenging conditions.
The results validate that our TRUST-UP algorithm not only satisfies input constraints and safety requirements $\C_{u,i}$ in~\eqref{C3}, $\C_{c,i}$ in~\eqref{C1} and $\C_{s,i}$ in~\eqref{C2} for all pursuit UAVs, but also improves the stability and robustness of target pursuit, making it suitable for scenarios with external disturbances and dynamic targets. {We further compared TRUST-UP CBF-QP with an input constrained CBF-QP technique~\cite{ames2019control} under identical simulation settings. TRUST-UP reduced the mean QP solving time from $0.695$ms to $0.597$ms ($14.1$\%), indicating improved online computational efficiency.}

%% 5. Conclusion
\section{Conclusion}
\label{sec-conclusion}
    In this paper, we proposed the TRUST-UP algorithm to address the aerial pursuit problem for autonomous UAVs, ensuring safety regarding collision avoidance, separation standards, and thrust constraints in urban aviation environments with dynamic human activities. The algorithm integrates three CBF constraints into a safety filter that transforms unsafe RL outputs into verifiably safe flight commands via a QP. A transparent switching strategy enhances feasibility while providing interpretable safety decisions, ensuring solutions satisfy the KKT conditions for all safety constraints. Through formal verification and simulations, we demonstrated that TRUST-UP achieves the safety guarantees and operational transparency required for airworthiness certification, addressing key trustworthiness requirements in aviation safety standards. 
% This work provides a certifiable framework for trustworthy AI deployment in safety-critical UAV operations, enabling autonomous aerial vehicles to operate safely in congested urban airspace while maintaining the explainable decision logic essential for regulatory compliance and public acceptance in future urban air mobility systems.
{While other constrained RL methods~\cite{dai2023augmented,xu2021crpo} also address safety through reward-cost trade-offs during training, the focus of this work is on an online CBF-QP safety filter, which provides a more transparent mechanism for multi-constraint handling in the considered human-aware pursuit scenario.} {In the present TRUST-UP framework, the human-aware safety requirement is modeled as a deterministic trust radius, so that the psychological safety zone can be converted into a hard constraint within the online CBF-QP safety filter. This choice is intended to prioritize transparent and real-time multi-constraint safety handling in complex pursuit scenarios. Future research will explore adaptive trust radius based on human movement and physiological markers to achieve more advanced human-aware navigation.}

\appendix
% Nomenclatures used in this work
\section{Appendix}
    The proof of Lemma~\ref{lem-kcbf3} is provided as follows. 
\begin{proof}
Let $\zeta_i = x_i - q_i$. The derivative of $h_{u,i}$ is given by 
\begin{equation}\begin{aligned}\label{temp2}
&\dot h_{u,i}= \dfrac{1}{2}(\kappa\dot \kappa-u_i^\top \dot u_i)\\
&\begin{multlined}[.85\linewidth]
=  -\dfrac{2\kappa (\zeta_i^\top\zeta_i-\ell^2)}{(\zeta_i^\top\zeta_i-\ell^2)^2+\epsilon)^2}\zeta_i^\top(f_i+g_iu_i-\dot q_i) \\
-\dfrac{2\kappa (\zeta_i^\top\zeta_i-\ell^2)}{(\zeta_i^\top\zeta_i-\ell^2)^2+\epsilon)^2}\zeta_i^\top Y_i\theta-u_i^\top Z_i\xi +u_i^\top v_i
\end{multlined}
\\
&\begin{multlined}[.85\linewidth]
=
-\dfrac{2\kappa (\zeta_i^\top\zeta_i-\ell^2)}{(\zeta_i^\top\zeta_i-\ell^2)^2+\epsilon)^2}\zeta_i^\top(f_i+g_iu_i-\dot q_i) + u_i^\top v_i\\- \dfrac{2\kappa (\zeta_i^\top\zeta_i-\ell^2)}{(\zeta_i^\top\zeta_i-\ell^2)^2+\epsilon)^2}\zeta_i^\top Y_i(\varepsilon_\theta+\hat \theta)\\-u_i^\top Z_i(\varepsilon_\xi + \hat \xi).
\end{multlined}
\end{aligned}\end{equation}
Using Assumption~\ref{asm-p dot p}, there always exists positive constant $\rho_v\in\R_{\geq 0}$ such that for all $t\geq 0$, $\norm{\dot q_i(t)}\leq \rho_v$. Therefore, we can yield the following inequality from~\eqref{temp2}
\begin{equation}\begin{aligned}\label{dot h3}
&\dot h_{u,i} \\
&\begin{multlined}[.85\linewidth]
\geq-\dfrac{2\kappa (\zeta_i^\top\zeta_i-\ell^2)}{(\zeta_i^\top\zeta_i-\ell^2)^2+\epsilon)^2}\zeta_i^\top(f_i+g_iu_i-\dot q_i) + u_i^\top v_i\\
- \dfrac{2\kappa (\zeta_i^\top\zeta_i-\ell^2)}{(\zeta_i^\top\zeta_i-\ell^2)^2+\epsilon)^2}\zeta_i^\top Y_i\hat \theta -u_i^\top Z_i\hat \xi
\\
-\norm{\dfrac{2\kappa (\zeta_i^\top\zeta_i-\ell^2)}{(\zeta_i^\top\zeta_i-\ell^2)^2+\epsilon)^2}\zeta_i^\top Y_i}\norm{\varepsilon_\theta} -\norm{u_i^\top Z_i}\norm{\varepsilon_\xi}\\
- \norm{\dfrac{2\kappa (\zeta_i^\top\zeta_i-\ell^2)}{(\zeta_i^\top\zeta_i-\ell^2)^2+\epsilon)^2}\zeta_i}\rho_v.
\end{multlined}
\end{aligned}\end{equation}
Define $\hbar_{u,i} = \dot h_{u,i} + \alpha(h_{u,i})$, 
using~\eqref{error bound} and \eqref{dot h3}, one yields
\begin{equation}\label{dot h3-2}
\begin{multlined}[.85\linewidth]
\hbar_{u,i} \geq -\dfrac{2\kappa (\zeta_i^\top\zeta_i-\ell^2)}{(\zeta_i^\top\zeta_i-\ell^2)^2+\epsilon)^2}\zeta_i^\top (f_i+g_iu_i + Y_i\hat \theta) 
\\
-\norm{\dfrac{2\kappa (\zeta_i^\top\zeta_i-\ell^2)}{(\zeta_i^\top\zeta_i-\ell^2)^2+\epsilon)^2}\zeta_i}\left(\norm{Y_i}\nu+\rho_v\right) \\ -u_i^\top Z_i\hat \xi -\norm{u_i^\top Z_i}\eta + u_i^\top v_i + \alpha(h_{u,i}).
\end{multlined}
\end{equation}
\end{proof}
The proof of Lemma~\ref{lem-kcbf1} is provided as follows. 
\begin{proof}
We begin with defining a sufficiently smooth function $\hbar_{c,i,k}$:
\begin{equation}\label{h bar}
\hbar_{c,i,k}=\ddot h_{c,i,k} + \iota_i\dot h_{c,i,k} + \alpha (h_{c,i,k})
\end{equation}
as a HOCBF with relative degree $r = 2$ of the CBFs $h_{c,i,k}$ in~\eqref{h1}. 
Then, to prove Lemma~\ref{lem-kcbf1}, one needs to show that $\hbar_{c,i,k}(t)\geq 0$ for all $t>0$ and all $k\in\mathcal{I}_k$, such that $h_{c,i,k}(t)\geq 0$ for all $t\geq 0$. This property holds if $\dot{\hbar}_{c,i,k}$ can be expressed in the form of (or larger than) $-\mu \dot{\hbar}_{c,i,k}$ where $\mu>0$ with $\dot{\hbar}_{c,i,k}(0)\geq 0$.

% We first calculate the derivative of $h_1$ here:
% \begin{equation}\label{dot h1}
% \begin{multlined}
% \dot h_1 = 2\zeta_i^\top (f + gu + Y\theta- \dot p_i),
% \end{multlined}
% \end{equation}
% \begin{equation}
% \begin{multlined}
% \ddot h_1 = 2x_{10}(\dot f_{10} + \dot g_{1}u_{1}+\dot g_0u_0)+ \\2(f_{10}+(g_{1}u_{1}-g_{0}u_{0})+Y_1\theta_1-Y_0\theta_0)^2\\ + 2x_{10}(\dot Y_1\theta_1-\dot Y_0\theta_0) - 2x_{10}g_{0}v_0\\ 
% + 2x_{10}(g_1Z_1\xi_1-g_0Z_0\xi_0)
% + 2x_{10}g_{1}v_1 
% \end{multlined}
% \end{equation}
% Denote $(g_{1}u_{1}-g_{0}u_{0})= g_{10}u_{10}$, $\dot g_{1}u_{1}+\dot g_0u_0= \dot g_{10}u_{10}$, and $g_1Z_1\xi_1-g_0Z_0\xi_0=g_{10}Z_{10}\xi_{10}$,  $Y_1\theta_1-Y_0\theta_0 = Y_{10}\theta_{10}$, and $\dot Y_1\theta_1-\dot Y_0\theta_0 = \dot Y_{10}\theta_{10}$, $Y_1(\hat \theta_1 + \varepsilon_\theta)-Y_0(\hat \theta_1 + \varepsilon_\theta) = Y_{10}(\hat \theta_{10} + \varepsilon_\theta)$, and $g_1Z_1(\hat \xi_1+\varepsilon_\xi)-g_0Z_0(\hat \xi_0+\varepsilon_\xi)_0=g_{10}Z_{10}(\hat \xi_{10}+\varepsilon_\xi)$.
By reduction and simplification, we yield
\begin{equation}\label{dot h1}
\begin{aligned}
\dot h_{c,i,k}
&\begin{multlined}
= 2\zeta_i^\top (f_i + g_iu_i + Y_i(\hat \theta +\varepsilon_\theta) - \dot p_k),
\end{multlined}\\
&\begin{multlined}
\geq 2\zeta_i^\top(f_i + g_iu_i) +  2\zeta_i^\top Y_i\hat \theta  - 2\norm{\zeta_i^\top Y_i}\norm{\varepsilon_\theta},
\end{multlined}
\end{aligned}
\end{equation}
\begin{equation}\begin{aligned}\label{ddot h1}
&\ddot h_{c,i,k}\\
&\begin{multlined}[.85\linewidth]
\geq 2\zeta_i^\top(\dot f_i + \dot g_i u_i) + 2\zeta_i^\top \dot Y_i\hat \theta - 2\norm{\zeta_i\dot Y_i}\norm{\varepsilon_\theta}
\\ - 2\norm{\zeta_i}\rho_a\!+\! 2\zeta_i^\top g_iZ_i\hat\xi \!-\!\norm{2\zeta_i^\top g_iZ_i}\norm{\varepsilon_\xi}  
\!-\!2\zeta_i^\top g_iv_i.
\end{multlined}
\end{aligned}\end{equation}
By using~\eqref{error bound} and applied Lemma~\ref{lem-updatelaw} to~\eqref{dot h1} and~\eqref{ddot h1}, we yield 
\begin{equation}\label{dot h1-2}
\begin{multlined}
\dot h_{c,i,k} 
\geq 2\zeta_i^\top(f_i + g_iu_i) + 2\zeta_i^\top Y_i\hat \theta - 2\norm{\zeta_i^\top Y_i}\nu,
\end{multlined}
\end{equation}
\begin{equation}\label{ddot h1-2}
\begin{multlined}[.85\linewidth]
\ddot h_{c,i,k} \geq 2\zeta_i^\top(\dot f_i + \dot g_iu_i) + 2\zeta_i^\top\dot Y_i\hat \theta -\norm{2\zeta_i^\top\dot Y_i} \nu \\- 2\norm{\zeta_i}\rho_a + 2\zeta_i^\top g_iZ_i\hat\xi -\norm{2\zeta_i^\top g_iZ_i}\eta
\\- 2\zeta_i^\top g_iv_i.
\end{multlined}
\end{equation}
Using~\eqref{dot h1-2} and~\eqref{ddot h1-2}, we yield
\begin{equation}\label{hbar>}
\begin{multlined}[.85\linewidth]
\hbar_{c,i,k} \geq 
2\zeta_i^\top(\dot f_i + \dot g_iu_i) + 2\zeta_i^\top\dot Y_i\hat \theta -\norm{2\zeta_i^\top\dot Y_i} \nu \\- 2\norm{\zeta_i}\rho_a + 2\zeta_i^\top g_iZ_i\hat\xi -\norm{2\zeta_i^\top g_iZ_i}\eta
\\- 2\zeta_i^\top g_iv_i\\
+ \iota_k \Big( 2\zeta_i^\top(f_i + g_iu_i) + 2\zeta_i^\top Y_i\hat \theta \\- 2\norm{\zeta_i^\top Y_i}\nu\Big) + \alpha(h_{c,i,k}(x_i)).
\end{multlined}
\end{equation}
\end{proof}

%% Acknowledgements
\section*{Acknowledgements}
This research is supported by the National Research Foundation, Singapore, and the Civil Aviation Authority of Singapore, under the Aviation Transformation Programme.
Sponsor Award Number: A-25-07692. Award Number: \#025059-00001. Title: Multi Sector Planner.

%% Data availability
\section*{Data Availability Statement}
Data made available at https://github.com/DengYaosheng/RLCBF.git.

\clearpage
\bibliographystyle{elsarticle-num-names}
\bibliography{VALIO}

\end{document}